\documentclass[journal,draftcls,onecolumn,12pt,twoside]{IEEEtranTCOM}

\normalsize

\usepackage{amsmath}
\usepackage{epic,xspace,epsfig,graphics,latexsym,times,syntonly}
\usepackage{theorem,amsmath,amsbsy,pifont,subfigure,verbatim,amsfonts,amssymb,enumerate,epstopdf}
\usepackage{psfrag,color}
\DeclareMathOperator*{\argmin}{argmin}
\DeclareMathOperator*{\argmax}{argmax}
\psfull
\makeatletter
\usepackage{lettrine}
\newtheorem{theorem}{Theorem}[section]

\newtheorem{proposition}[theorem]{Proposition}

\usepackage{algorithm}
\usepackage{algpseudocode}
\newenvironment{proof}[1][Proof]{\begin{trivlist}
\item[\hskip \labelsep {\bfseries #1}]}{\end{trivlist}}

\newenvironment{remark}[1][Remark:]{\begin{trivlist}
\item[\hskip \labelsep {\it\bfseries #1}]}{\end{trivlist}}

\begin{document}

\title{Efficient Decoding Algorithms for the Compute-and-Forward Strategy}

\author{Asma~Mejri,~\IEEEmembership{Student Member,~IEEE,}
        Ghaya~Rekaya,~\IEEEmembership{Member,~IEEE}
\thanks{A. Mejri and G. Rekaya are with the Communications and Electronics Department
of Telecom-ParisTech, Paris, 75013, France. Emails: amejri,rekaya@telecom-paristech.fr.}}

\markboth{IEEE Transactions on Communications}%
{Submitted paper}

\maketitle

\begin{abstract}
 We address in this paper decoding aspects of the Compute-and-Forward (CF) physical-layer network coding strategy. It is known that the original decoder for the CF is asymptotically optimal. However, its performance gap to optimal decoders in practical settings are still not known. In this work, we develop and assess the performance of novel decoding algorithms for the CF operating in the multiple access channel. For the fading channel, we analyze the ML decoder and develop a novel diophantine approximation-based decoding algorithm showed numerically to outperform the original CF decoder. For the Gaussian channel, we investigate the \textit{maximum a posteriori} (MAP) decoder. We derive a novel MAP decoding metric and develop practical decoding algorithms proved numerically to outperform the original one.
\end{abstract}


\IEEEpeerreviewmaketitle

\section{Introduction}

\IEEEPARstart{L}{ast} few years have witnessed the emergence of a very promising linear physical-layer network coding protocol termed \textit{Compute-and-Forward}. Introduced by Nazer and Gastpar in \cite{Nazer08}, this scheme allows to harness the multiple access interference to achieve higher transmission rates. This new framework is applicable to any network configuration accomodating source nodes, relays and destinations that communicate through linear additive white Gaussian noise channels. The role of a relay node observing the output of a multiple access channel is to decode a \textit{linear integer} combination of source codewords. Given enough linear equations, the end destination in the network can ideally recover the original source messages with high transmission rates thanks to the potential properties of nested lattice codes. 

The original decoder for the CF consists of a scaling operation and a minimum distance decoding. Under these assumptions, a union bound estimate of the error probability at the relays was derived in \cite{Silva_isit10}, and we have addressed in \cite{Asma1} and \cite{Asma2} the end-to-end error performance evaluation in the multi-source relay channel and the two-way relay channel respectively. Later on, independently in \cite{AsmaGhaya} and \cite{Belfiore12}, the Maximum Likelihood (ML) decoder was investigated. 
 An algebraic extension of the CF using lattice partitions related to  finitely generated modules over principal ideal domains was proposed by Feng \textit{et al.} in \cite{Silva_isit10} assuming also minimum distance decoding. These works focus on the information theoretic performance of the CF and demonstrate that minimum distance decoding is asymptotically optimal. Nevertheless, the performance of this decoder in practical settings (finite lattice dimensions and low-complexity encoding schemes) and the gap to optimal decoders are still not known. We try in this work to find answers to these issues by investigating optimal decoding criteria for the CF in the basic multiple access channel. After reviewing the original CF encoding and decoding schemes in section II, our contributions come into light as follows: in section III, we investigate efficient ML decoding for the fading channel considering integer lattices. By analyzing the one-dimensional case, we develop a novel near-ML decoder based on diophantine approximation and show by numerical results its gain over the original CF decoder for $\mathbb{Z}-$lattices. In section IV, we analyze the MAP decoder for the CF in the Gaussian channel using real-valued lattices. We derive a novel union bound estimate on the error probability and propose a lattice design criterion. Then, we derive a novel MAP decoding metric based on which we develop novel efficient decoding algorithms proved numerically to outperform the conventional CF decoder while maintaining the same complexity order. 

\section{Compute-and-Forward in Basic MAC: Original Work}

\subsection{Preliminaries on Lattice Coding}
An $n$-dimensional \textit{\textbf{lattice}} $\Lambda$ is a set of points of $\mathbb{R}^{n}$ given by $\Lambda=\left\lbrace \mathbf{x}= \mathbf{M} \mathbf{s} , \mathbf{s} \in \mathbb{Z}^{n} \right\rbrace$ where $\mathbf{M}$ is called a generator matrix of the lattice. The main characteristic of $\Lambda$ is \textit{linearity}, i.e. for any $a, b \in \mathbb{Z}$ and $\mathbf{x}, \mathbf{y} \in \Lambda$, $a\mathbf{x}+b\mathbf{y} \in \Lambda$. 

A \textit{\textbf{lattice quantizer}} $\mathit{Q}_{\Lambda}$ satisfies for $\mathbf{x} \in \mathbb{R}^{n}$, $\mathit{Q}_{\Lambda}(\mathbf{x}) = \argmin_{\lambda \in \Lambda} \parallel \mathbf{x} - \lambda \parallel$. The set of points that quantize to a given lattice point is called the \textit{\textbf{Voronoi Region}}. The \textit{\textbf{fundamental Voronoi Region}} $\mathcal{V}_{\Lambda}$ of a lattice $\Lambda$ corresponds to the voronoi region of the zero vector. The modulo operation returns the quantization error with respect to $\Lambda$. For $\mathbf{x} \in \mathbb{R}^{n}$: $\left[ \mathbf{x} \right]\mathrm{mod}\Lambda = \mathbf{x} -  \mathit{Q}_{\Lambda} \left( \mathbf{x} \right)$. 
 
 A \textit{\textbf{nested lattice code}} $\mathcal{C}_{\Lambda}$ is the set of all points of a lattice $\Lambda_{\mathrm{F}}$ (termed the \textit{Fine} lattice) that fall within the fundamental Voronoi region of a lattice $\Lambda_{\mathrm{C}}$ (termed the \textit{Coarse} lattice) as: $\mathcal{C}_{\Lambda}= \left\lbrace \mathbf{\lambda} = \left[ \lambda_{\mathrm{F}} \right]\mathrm{mod}\Lambda_{\mathrm{C}}, \lambda_{\mathrm{F}} \in \Lambda_{\mathrm{F}}\right\rbrace$.
 
\subsection{System Model and Assumptions}
We consider the real-valued fading Multiple Access Channel (MAC) composed of $N$ sources $\mathrm{S}_i,i=1,...,N$ and a common receiver. Extension of our results to the complex-valued channel follows by considering the real and imaginary parts of the channel outputs separately. Source $\mathrm{S}_i$  delivers a length-$k$ finite field message $\mathbf{w}_i \in \mathbb{F}_{p}^{k}$
drawn independently and uniformally. Encoders $\mathcal{E}$ at the sources implement the same mapping $\phi$ to map the messages $\mathbf{w}_i$ onto codewords $\mathbf{x}_i$ from the same nested lattice code $\mathcal{C}_{\Lambda}$ according to the symmetric power constraint given by:
\begin{equation}\label{power}
\frac{1}{n} \mathbb{E} \left( \parallel \mathbf{x}_{i} \parallel^{2} \right) \leq P ~,~P > 0
\end{equation}
$\Lambda_{\mathrm{F}}$ corresponds to the coding lattice and $\Lambda_{\mathrm{C}}$ acts to satisfy the power constraint $P$. The codewords are assumed to be independent and uniformally distributed over $\mathcal{C}_{\Lambda}$. The message rate is equal to $r=\frac{k}{n} \log{p}$ and is the same for all sources. Sources transmit then their codewords simultaneously across the channel. The received vector is written as:
\begin{equation}
\mathbf{y} = \sum_{i=1}^{N} h_{i} \mathbf{x}_i + \mathbf{z}
\end{equation}
where $h_i \in \mathbb{R}$ denotes the fading coefficient from source $S_i$ to the receiver and $\mathbf{z} \in \mathbb{R}^{n}$ denotes the additive white Gaussian noise of zero-mean and variance $\sigma^{2}$. Let $\mathbf{h}=\left[h_1,...,h_N \right]^{t}$ denote the channel coefficients vector. In this work we assume fixed channel vector. Nevertheless, our results hold for the fast fading and slow fading channels. We assume also that channel state information (CSI) is available only at the receiver and denote by $\rho = \frac{P}{\sigma^2}$ the signal-to-noise ratio (SNR).

\subsection{Decoding Scheme for the Compute-and-Forward}
The receiver attempts to decode a noiseless integer linear combination in the form:
\begin{equation}
\lambda = \left[ \sum_{i=1}^{N} a_i \mathbf{x}_i \right] \mathrm{mod}~\Lambda_{\mathrm{C}} ~,~a_i \in \mathbb{Z}, i=1,...,N 
\end{equation}
Where the network code vector $\mathbf{a}=[a_1,...,a_N]^{t} \in \mathbb{Z}^{N}$ is chosen by the receiver. 
The latter is equipped with a decoder $\mathcal{D} : \mathbb{R}^{n} \rightarrow \Lambda$, that recovers an estimate $\hat{\lambda}$ of $\lambda$. A decoding error occurs if $\hat{\lambda} \neq \lambda$ and the desired equation with a coefficient vector $\mathbf{a}$ is decoded with an average probability of error $\epsilon$ if $\hat{\lambda} \stackrel{\triangle}{=} \mathcal{D}\left( \mathbf{y} \right)$ and $\mathrm{Pr} \left( \hat{\lambda} \neq \lambda \right) < \epsilon$. 
A computation rate $\mathcal{R}(\mathbf{h},\mathbf{a})$ is said to be achievable if for any $\epsilon > 0$ and $n$ large enough, there exist an encoder $\mathcal{E}$ and a decoder $\mathcal{D}$, such that for any channel fading vector $\mathbf{h} \in \mathbb{R}^{N}$ and network code vector $\mathbf{a} \in \mathbb{Z}^{N}$, the receiver can recover the desired equation with an average probability of error $\epsilon$ as long as the source message rate $r$ satisfies: $r < \mathcal{R}\left( \mathbf{h},\mathbf{a} \right)$. 

The receiver selects a scalar $\alpha \in \mathbb{R}$ and an integer vector $\mathbf{a}$ and performs the following steps:
\begin{enumerate}
\item Scale the channel output: $\tilde{\mathbf{y}}= \alpha \mathbf{y}=\sum_{i=1}^{N} a_i \mathbf{x}_i + \sum_{i=1}^{N} \left( \alpha h_i-a_i \right) \mathbf{x}_i +\alpha \mathbf{z}$. The resulting effective noise $\mathbf{z}_{\mathrm{eq}}=\sum_{i=1}^{N} \left( \alpha h_i-a_i \right) \mathbf{x}_i +\alpha \mathbf{z}$ is not Gaussian since composed of a quantization error involving the original codewords. At this level, $\mathbf{t}=\sum_{i=1}^{N} a_i \mathbf{x}_i \in \Lambda_{\mathrm{F}}$. 
 \item Decode to the nearest point in the fine lattice: $\hat{\mathbf{t}} = \mathnormal{Q}_{\Lambda_{\mathrm{F}}}\left(\tilde{\mathbf{y}} \right)$.
 \item Take the modulo operation with respect to the coarse lattice: $\hat{\lambda}= \left[ \hat{\mathbf{t}}\right]\mathrm{mod}~\Lambda_{\mathrm{C}}$.
\end{enumerate}
We summarize the results regarding the CF protocol in the following theorems \cite{Nazer08}.
\begin{theorem}[Computation rate]\label{Rcomp}
For a real-valued MAC with channel vector $\mathbf{h}$, and network code vector $\mathbf{a} \in \mathbb{Z} ^{N}$ the following computation rate $\mathnormal{R}_{\mathrm{comp}}$, for $\alpha \in \mathbb{R}$, is achievable:
\begin{align}\label{rateCF}
R_{\mathrm{comp}}(\mathbf{h},\mathbf{a} ) &= \frac{1}{2} \log ^{+}{ \left( \frac{\rho}{\alpha ^2 + \rho \parallel \alpha \mathbf{h} - \mathbf{a} \parallel ^{2} } \right)}
\end{align}
where $\log^{+}(x)=max(\log(x),0)$.
\end{theorem}
\begin{theorem}[Optimal scaling factor]
The computation rate given in Theorem \ref{Rcomp} is only maximized for the MMSE scaling factor $\alpha_{\mathrm{opt}}$ given by: $\alpha_{\mathrm{opt}} = \frac{\rho \mathbf{h}^{t} \mathbf{a} }{1+\rho \parallel \mathbf{h} \parallel ^{2}}$.
\end{theorem}

\begin{theorem}[Optimal network code vector]\label{aopt}
The optimal network code vector satisfies:
\begin{equation}\label{opti}
\mathbf{a}_{\mathrm{opt}} = \argmin_{\mathbf{a} \in \mathbb{Z}^{N}, \mathbf{a} \neq \mathbf{0}} \left \{ \mathbf{a}^{t}\mathbf{G} \mathbf{a} \right \}
\end{equation}
where $\mathbf{G} = \mathbf{I}_{N} - \frac{\rho}{1+\rho\parallel \mathbf{h} \parallel ^{2}}\mathbf{h}\mathbf{h}^{t}$ is definite positive. $\mathbf{a}_{\mathrm{opt}}$ corresponds to the shortest vector in the lattice $\Lambda_{\mathbf{G}}$ of Gram matrix $\mathbf{G}$.
\end{theorem}

The conventional decoding scheme for the CF consists of an MMSE scaling operation and a minimum distance decoding. The problem is that in presence of the non-Gaussian effective noise $\mathbf{z}_{\mathrm{eq}}$, minimum distance decoding is not ML decoding. Although the conventional decoder is proved to be optimal in asymptotic regime using high dimensional lattices, its performance gap to the optimal decoders is not known particularly in practical settings using finite-dimensional lattices. We aim in the following to study optimal decoding criteria and develop practical efficient decoding algorithms. Although we will consider the real-valued channel, our results hold in the complex-valued channel case using the same techniques at the real and imaginary parts of the channel output separately.

\section{Efficient Decoders in Fading Channels}
We start in this section with the case of fading channels. The tools we will use in our analysis are valid only in the case of integer lattices, thus we will consider an $n-$dimensional nested lattice code $\mathcal{C}_{\Lambda} \subset \mathbb{Z}^{n}$ involving a fine lattice $\Lambda_{\mathrm{F}} \subset  \mathbb{Z}^{n}$ of a generator matrix $\mathbf{M}$ and a coarse lattice $\Lambda_{\mathrm{C}} \subset \mathbb{Z}^{n} $. For this case, $\mathbf{M}$ is an integer full rank matrix. We will start with the multi-dimensional case which was independently studied in \cite{Belfiore12} then we provide more in depth analysis regarding the one-dimensional case. 
\subsection{Problem Statement}
After selecting $\alpha$ and $\mathbf{a}$, the receiver scales the channel output to get:
\begin{equation}
\tilde{\mathbf{y}} = \sum_{i=1}^{N} a_i \mathbf{x}_i + \sum_{i=1}^{N} \left( \tilde{h}_i - a_i \right) \mathbf{x}_i + \tilde{\mathbf{z}}
\end{equation}
where $\tilde{h}_i=\alpha h_i, i=1,...,N$ and $\tilde{\mathbf{z}}=\alpha \mathbf{z}$, and attempts to decode $\lambda = \left[ \sum_{i=1}^{N} a_i \mathbf{x}_i \right] \mathrm{mod}~\Lambda_{\mathrm{C}}$. We are concerned in this part with ML decoding for recovering the integer combination $\mathbf{t}=\sum_{i=1}^{N} a_i \mathbf{x}_i$. The modulo-lattice operation is performed in a second stage separately and does not impact the error performance. Thus, we evaluate the decoding error probability defined as: $\mathrm{P}_{\mathrm{e}} = \mathrm{Pr}\left( \hat{\mathbf{t}} \neq \mathbf{t} \right)$. 
Given the vector $\mathbf{a}$ and the shaping boundaries for the original codewords, it is known that the searched vector $\mathbf{t}$ belongs to a subset $\Lambda_{f}$ in the fine lattice $\Lambda_{\mathrm{F}}$. This shaping constraint is also disregarded under the conventional CF decoder. In the following, we analyze the ML decoder that takes into consideration this shaping condition. 

\subsection{ML Decoding Metric}
The ML criterion is based on maximizing the conditional probability $p\left( \tilde{\mathbf{y}} | \mathbf{t} \right)$ according to:
 \begin{equation}\label{ml}
 \hat{\mathbf{t}} = \argmax_{\mathbf{t} \in \Lambda_{f}}  p\left( \tilde{\mathbf{y}} | \mathbf{t} \right) 
 \end{equation}
Given that $\mathbf{t} = \sum_{i=1}^{N} a_i \mathbf{x}_i$, we can equivalently write (\ref{ml}) as:
\begin{equation}\label{ml_1}
 \small{\mathbf{\hat{t}} = \argmax_{\mathbf{t} \in \Lambda_{f}} \sum_{\scriptstyle \left ( \mathbf{x}_1,...,\mathbf{x}_N \right ) \in \Lambda^{N}/\scriptstyle \sum_{i=1}^{N} a_i \mathbf{x}_i=\mathbf{t} }  p \left ( \tilde{\mathbf{y}} | \left ( \mathbf{x}_1,...,\mathbf{x}_N \right ) \right ) p \left ( \mathbf{x}_1,...,\mathbf{x}_N\right )} 
 \end{equation}
The transmitted codewords are assumed to be uniformally distributed over the nested lattice code $\mathcal{C}_{\Lambda}$, i.e., $\mathbf{x}_1,...,\mathbf{x}_N$ are equiprobable. On the other hand, we have,
\begin{equation}\label{proba}
p \left ( \tilde{\mathbf{y}} | \mathbf{x}_1,...,\mathbf{x}_N \right ) \hspace{0.4cm} \propto \exp \left ( \frac{-1}{2\tilde{\sigma}^{2}}\parallel \tilde{\mathbf{y}}- \sum_{i=1}^{N} \tilde{h}_i\mathbf{x}_i \parallel ^{2} \right )
\end{equation}
where $\tilde{\sigma}^{2}=\alpha^{2}\sigma^{2}$. Combining (\ref{proba}) and (\ref{ml_1}), we get:
 \begin{equation}\label{decoding}
  \small{\mathbf{\hat{t}}=\argmax_{\mathbf{t} \in \Lambda_{f}} \sum_{\scriptstyle \left ( \mathbf{x}_1,...,\mathbf{x}_N \right ) \in \Lambda^{N}/\scriptstyle  \sum_{i=1}^{N} a_i \mathbf{x}_i=\mathbf{t}} \exp \left ( \frac{-1}{2\tilde{\sigma}^{2}}\parallel \tilde{\mathbf{y}}- \sum_{i=1}^{N} \tilde{h}_i\mathbf{x}_i \parallel ^{2} \right )} 
 \end{equation}
Let 
\begin{equation}
\varphi \left ( \mathbf{t} \right ) = \sum_{\scriptstyle \left ( \mathbf{x}_1,...,\mathbf{x}_N \right ) \in \Lambda^{N}/\scriptstyle  \sum_{i=1}^{N} a_i \mathbf{x}_i=\mathbf{t} } \exp \left ( \frac{-1}{2\tilde{\sigma}^{2}}\parallel \tilde{\mathbf{y}}- \sum_{i=1}^{N} \tilde{h}_i\mathbf{x}_i \parallel ^{2} \right )
\end{equation}
Our objective in the following is to express $\varphi$ as a function of the desired equation $\mathbf{t}$. To this end, we need to express the codewords $\mathbf{x}_i,i=1,...,N$ as functions of $\mathbf{t}$. Given the integer nature of the vector $\mathbf{a}$ and the codewords $\mathbf{x}_i$, this task requires to solve the system of diophantine equations $\mathbf{t} = \sum_{i=1}^{N} a_i \mathbf{x}_i$. For $n-$dimensional vectors this can be done using the Hermite Normal Form (HNF) of integral matrices \cite{Cohen93,Lazebnik96} as explained in the following.

Define the integer-valued matrix $\tilde{\mathbf{M}} \in \mathbb{Z}^{n \times nN}$ as $\tilde{\mathbf{M}} = \left[ a_1\mathbf{M} ~ a_2\mathbf{M}~...~a_N\mathbf{M} \right]$.
The Hermite Normal Form of $\tilde{\mathbf{M}}$ is such that: $\tilde{\mathbf{M}} \mathbf{U} = \left[ \mathbf{0}^{n \times (N-1)n} | \mathbf{B} \right]$
where $\mathbf{U} \in \mathbb{Z}^{nN \times nN}$ is a unimodular matrix, and $\mathbf{B} \in \mathbb{Z}^{n \times n}$ is an invertible matrix.
Then, we decompose the matrix $\mathbf{U}$ in the form:
\begin{equation}
\mathbf{U} = \left[ \begin{array}{cc}
\mathbf{U}_1 & \mathbf{V}_1 \\
\mathbf{U}_2 & \mathbf{V}_2 \\
\vdots & \vdots \\
\mathbf{U}_N & \mathbf{V}_N
\end{array} \right] ~,~\mathbf{V}_i \in \mathbb{Z}^{n \times n}~,~\mathbf{U}_i \in \mathbb{Z}^{n \times n(N-1)}
\end{equation}
The solution of the system of diophantine equations is then given by $\mathbf{x}_i = \mathbf{d}_i + \mathbf{v}_i$
where $\mathbf{v}_i=\mathbf{M}\mathbf{V}_i\mathbf{B}^{-1}\mathbf{t}$ and $\mathbf{d}_i$ belong to the lattice of a generator matrix $\mathbf{M}\mathbf{U}_i$ for $i=1,...,N$.
\subsection{Likelihood Function}
We go back now to the ML decoding rule defined in (\ref{decoding}) and replace the vectors $\mathbf{x}_i$ by the solution of the diophantine equations we obtain  
 \begin{equation}\label{mldecoding}
  \small{\mathbf{\hat{t}}=\argmax_{\mathbf{t} \in \Lambda_{f}} \sum_{\mathbf{q} \in \mathcal{L}} \exp \left ( \frac{-1}{2\tilde{\sigma}^{2}}\parallel \omega(\mathbf{t})- \mathbf{q} \parallel ^{2} \right )} 
 \end{equation}
where $\mathbf{q}=\sum_{i=1}^{N} \tilde{h}_i \mathbf{d}_i$ belongs to the lattice $\mathcal{L}$ of a generator matrix $\sum_{i=1}^{N} \tilde{h}_i \mathbf{M} \mathbf{U}_i$ and $\omega(\mathbf{t})=\tilde{\mathbf{y}}-\sum_{i=1}^{N} h_i \mathbf{M}\mathbf{V}_i \mathbf{B}^{-1}\mathbf{t}$. 

To find the ML solution, we need to maximize the likelihood function:
\begin{equation}
\varphi (\mathbf{t}) = \sum_{\mathbf{q} \in \mathcal{L}} \exp \left ( \frac{-1}{2\tilde{\sigma}^{2}}\parallel \omega(\mathbf{t})- \mathbf{q} \parallel ^{2} \right ) 
\end{equation}
This function is a sum of Gaussian measures, it is periodic and depends on the Signal-to-Noise Ratio. Additionally, its most important characteristic is that it can be flat, which means that for some values of the channel coefficients, the network code vector and the Signal-to-Noise Ratio, the maximum of $\varphi$ can be achieved by several values of $\mathbf{t}$, which makes the ML decision ambiguous and results in decoding errors. This flatness behavior is characterized by Belfiore and Ling in \cite{Belfiore12} by the so called the \textit{Flatness Factor}. For the ML decoding rule, we should minimize the flatness factor of the lattice $\mathcal{L}$ over which is performed the sum of the Gaussian measures in order to be able to distinguish the maximum values of the likelihood function and perform a correct decoding decision. Solving the ML decoding metric requires more research on the sum of Gaussian measures. Alternatively, authors in \cite{Belfiore12} propose an approximation of ML decoding based on \textit{Diophantine Approximation} and consists in the optimization problem given by:
\begin{equation}
\hat{\mathbf{t}} = \argmax_{\mathbf{t} \in \Lambda_{f},~ \mathbf{q} \in \mathcal{A}_{\mathcal{L}}} \parallel \omega(\lambda)- \mathbf{q} \parallel ^{2}
\end{equation}
Where $\mathcal{A}_{\mathcal{L}}$ is a finite subset of the lattice $\mathcal{L}$ fixed by the boundaries of the original codewords according to the transmission power constraint. For one-dimensional lattices, there are several algorithms pertaining to the resolution of the diophantine approximations of reals. However, solving the multi-dimensional case requires additionally to develop efficient algorithms to handle simultaneous diophantine approximations. We study in the following the 1-D case in more details.

\subsection{1-D Lattices Case Study}
Here, we focus on the case of 1-D lattices in $\mathbb{Z}$ and $N=2$. Transmitted codewords $x_1$ and $x_2$ are just integer scalars drawn i.i.d from the integer constellation over $\mathbb{Z}$ defined by $\mathcal{A}=\left[-S_\mathrm{m}~S_\mathrm{m} \right]$ for $S_\mathrm{m} \in \mathbb{Z}^{+}$. This integer codebook can be seen as a nested lattice code in $\mathbb{Z}$ involving the fine lattice $\Lambda_{\mathrm{F}}=\mathbb{Z}$ and the coarse lattice $\Lambda_{\mathrm{C}}=2S_\mathrm{m} \mathbb{Z}$. The channel output in this case is given by: $y = h_1 x_1 + h_2 x_2 + z$, with $h_i \in \mathbb{R}$ and $z \sim \mathcal{N}(0,\sigma^{2})$. The receiver selects the optimal scaling parameter and the optimal network code vector $\mathbf{a}=[a_1~a_2]^{t}$ and attempts to decode the integer combination $t=a_1 x_1+a_2 x_2$ from the integer set $\mathcal{A}_{t}$ determined by $S_\mathrm{m}$ and the values of the coefficients $a_1$ and $a_2$. The scaled channel output is given by:
\begin{align}
\tilde{y} &= a_1 x_1 + a_2x_2 + \left( \tilde{h}_1 - a_1 \right) x_1 + \left( \tilde{h}_2 - a_2 \right) x_2 + \tilde{z} \notag \\
&\tilde{h}_i=\alpha h_i, i=1,2 ~;~ \tilde{z}=\alpha z
\end{align}
Under these settings, the ML solution is given by:
 \begin{align}\label{mlscalar}
 \hat{t} &= \argmax_{t \in \mathcal{A}_{t}} \sum_{\scriptstyle \left (x_1,x_2 \right) \in \mathcal{A}^{2}/\scriptstyle a_1x_1+a_2x_2=t} \exp \left ( \frac{-1}{2\tilde{\sigma}^{2}}\parallel \tilde{y}-  \tilde{h}_1x_1 - \tilde{h}_2x_2 \parallel ^{2} \right ) 
 \end{align}
And the likelihood function is given by:
\begin{equation}
\varphi(t) = \sum_{\scriptstyle \left (x_1,x_2 \right) \in \mathcal{A}^{2}/\scriptstyle a_1x_1+a_2x_2=t} \exp \left ( \frac{-1}{2\tilde{\sigma}^{2}}\parallel \tilde{y}-  \tilde{h}_1x_1 - \tilde{h}_2x_2 \parallel ^{2} \right )
\end{equation}
Our aim now is to express $\varphi$ as a function of $t$ only. Therefore, we need to solve the \textit{Diophantine Equation} $t=a_1x_1+a_2x_2$. Let $g=a_1\wedge a_2$ denote the greatest common divisor (gcd) of $a_1$ and $a_2$. If the desired scalar $t$ is a multiple of $g$, the diophantine equation admits an infinite number of solutions in the form:
\begin{gather}\label{bounds}
\raisetag{-10pt}
\begin{cases}
 x_1 = \frac{u_1}{g} t + \frac{a_2}{g} k
\\
x_2 = \frac{u_2}{g} t - \frac{a_1}{g} k
\end{cases}  
\end{gather} 
where $k \in \mathbb{Z}$ and $\left(u_1,u_2\right)$ is a particular solution of the equation $a_1x_1+a_2x_2=g$ that can be derived using the \textit{Extended Euclid Algorithm} \cite{Cormen09}. If $t$ is not a multiple of $g$, then the diophantine equation has no solutions. For what concerns our case, the network code vector $\mathbf{a}$ corresponds to the coordinates of a lattice shortest vector, then the coefficients $a_1$ and $a_2$ are coprime. Thus, the diophantine equation under question has always infinite solutions given by the system in (\ref{bounds}) with $g=1$. Accordingly, we can write the ML solution in (\ref{mlscalar}) as
\begin{equation}\label{topt}
\hat{t} = \argmax_{t \in \mathcal{A}_t} \underbrace{\sum_{k=-\infty }^{+\infty } exp \left ( \frac{-1}{2\tilde{\sigma}^{2}} \parallel \tilde{y} - \gamma t + \beta k  \parallel ^{2} \right )}_{\varphi(t)}
\end{equation}
where $\gamma = \tilde{h}_1 u_1 + \tilde{h}_2 u_2$, $\beta = a_1 \tilde{h}_2 - a_2 \tilde{h}_1$ and $k \in \mathbb{Z}$.
   
\subsubsection{Properties of the likelihood function}
 $\varphi$ is a sum of gaussian functions, it is periodic with \textit{mean} $m = \tilde{y}$, \textit{period} $p = \frac{\beta}{2\tilde{\sigma}^{2}}$ and \textit{width} $w = \frac{\gamma}{2\tilde{\sigma}^{2}}$.  In addition, $\varphi$ depends on the SNR, the channel coefficients, the coefficient vector $\mathbf{a}$ and obviously on the constellation bounds defined by $S_m$.  We illustrate in Fig.\ref{flat1} an example of the likelihood function obtained for $S_m=5$, $x_1=3$, $x_2=4$ at $\mathrm{SNR}=10\mathrm{dB}$ and $\mathbf{h}=[-1.191~1.189]^{t}$. The optimal network code vector for this case is equal to $\mathbf{a}=[-1~1]^{t}$. Accordingly, the desired combination should be equal to $t=1$. The corresponding likelihood function depicted in Fig.\ref{flat1} is well maximized at $\hat{t}=1$. In this case, it is easy to decode the maximum of $\varphi(t)$ since we can distinguish a peak corresponding to the unique $\hat{t}$ for which this function is maximized.
 
 \begin{figure}[htp]
  \centering
  \subfigure[$S_\mathrm{m}=5,\mathrm{SNR}=60\mathrm{dB}$]{\includegraphics[height=6cm,width=5cm]{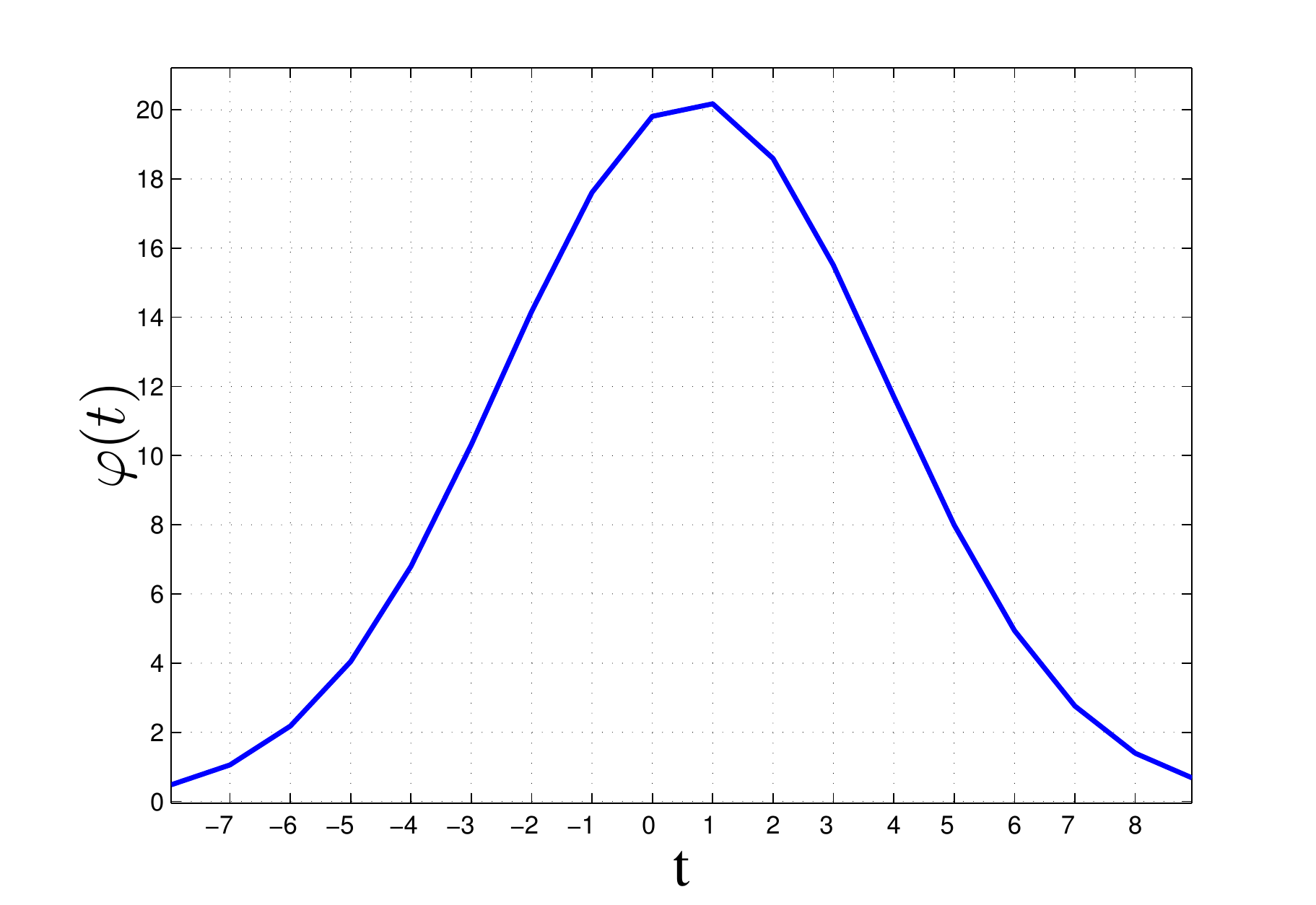}\label{flat1}} \
  \subfigure[$S_\mathrm{m}=5,\mathrm{SNR}=60\mathrm{dB}$]{\includegraphics[height=6cm,width=5cm]{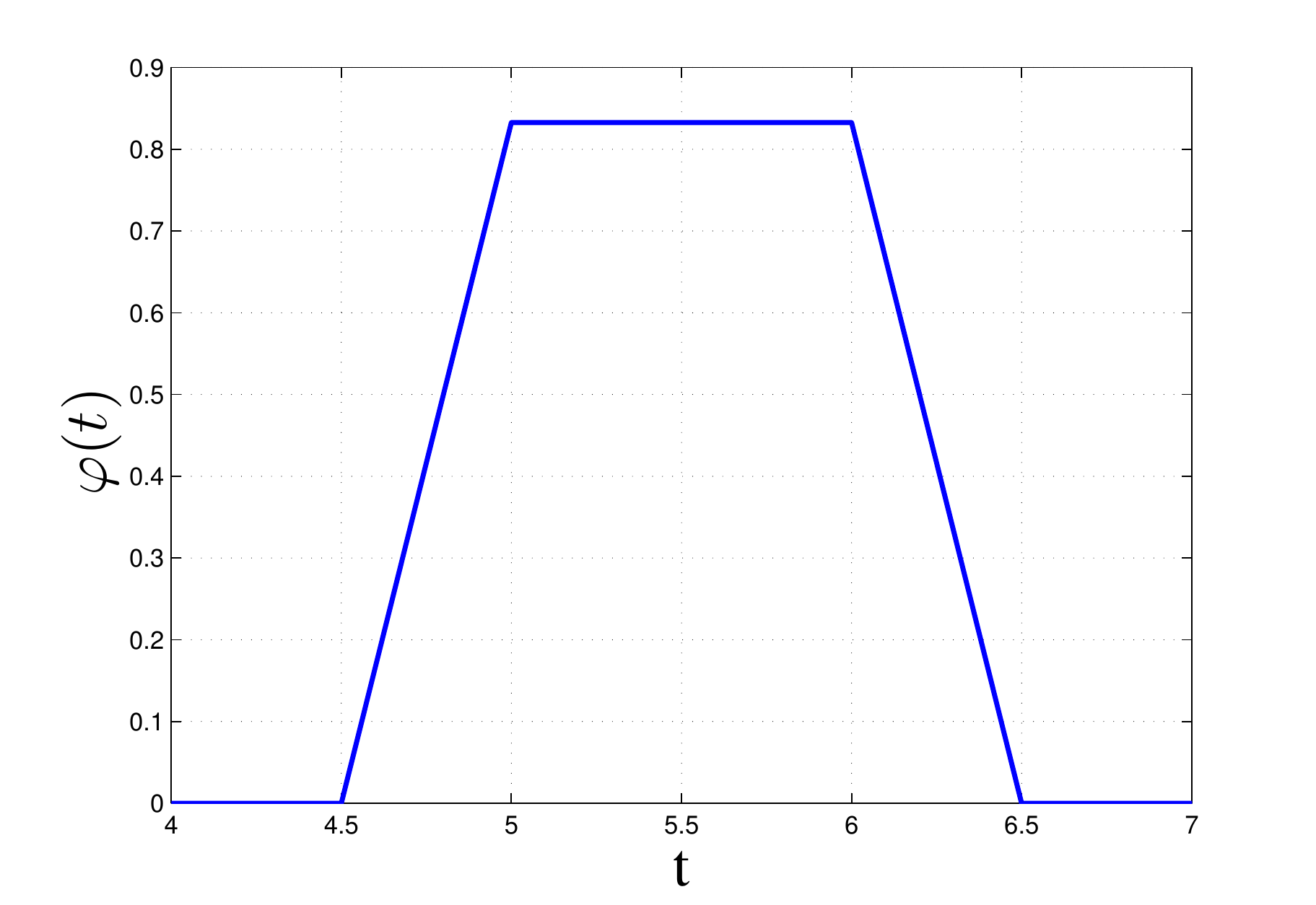}\label{flat2}} \
  \subfigure[$S_\mathrm{m}=10,\mathrm{SNR}=10\mathrm{dB}$]{\includegraphics[height=6cm,width=5cm]{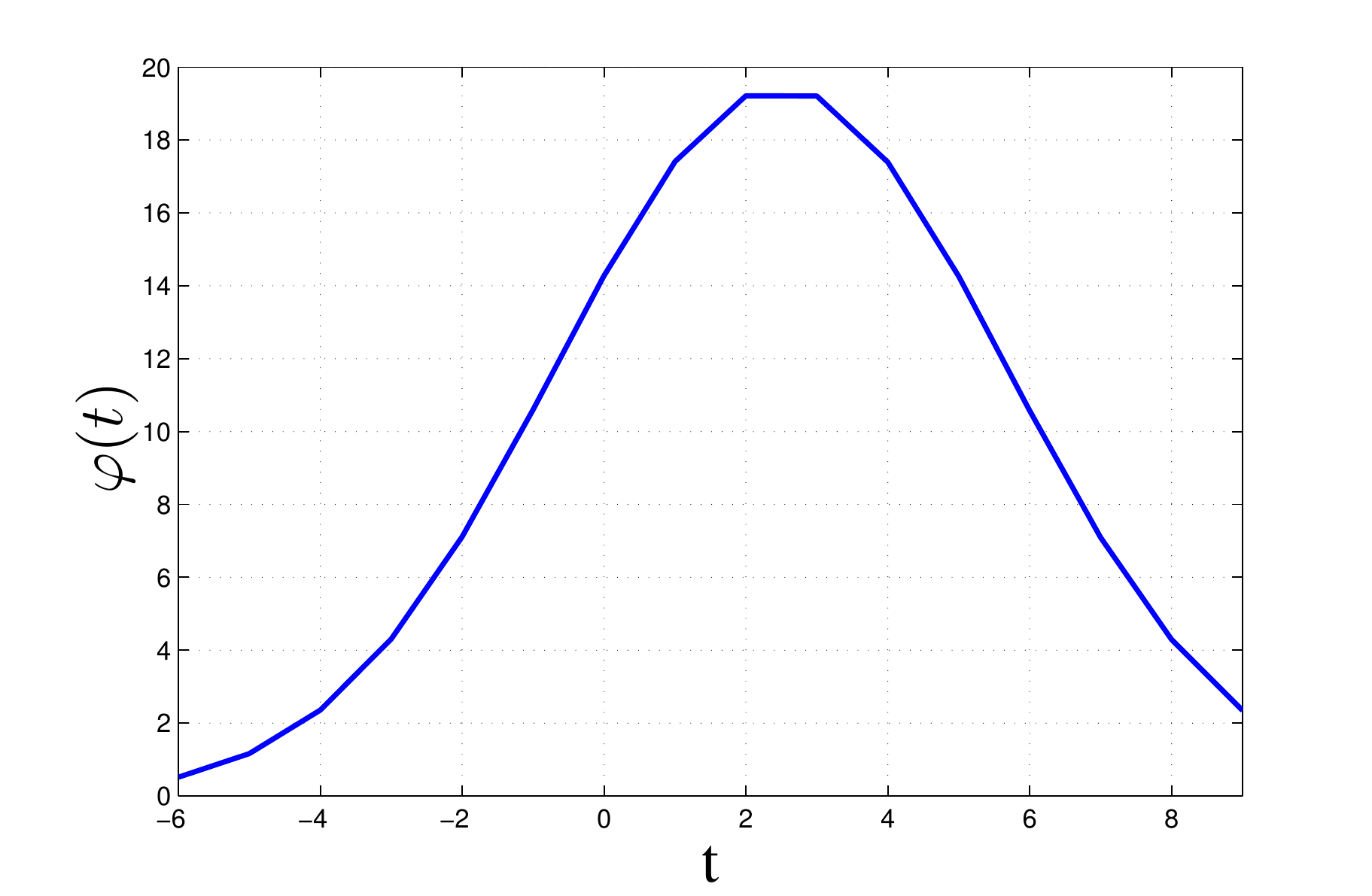}\label{flat3}}
  \caption{Examples of the Likelihood function.}
  \end{figure}
\emph{- Impact of $\mathbf{h}$ and $\mathbf{a}$:} the choice of the network code vector $\mathbf{a}$ can greatly impact the behavior of the likelihood function. Particularly, when this integer vector is aligned to the channel vector $\mathbf{h}$ (they become colinear), the period $p=\frac{a_1\tilde{h}_1-a_2\tilde{h}_2}{2\tilde{\sigma}^{2}}$ of the likelihood function becomes small and results in a flatness of $\varphi$ and impossibility of decoding the right $\hat{t}$ since the maximum can be obtained for different values. This result is demonstrated through Fig.\ref{flat2} obtained at $\mathrm{SNR}=60\mathrm{dB}$
$S_m=5, x_1=-5, x_2=-4, \mathbf{h}=[1.3681 -0.2359]^{t}, \mathbf{a}=[-1~0]^{t}$. The maximum of the likelihood function is obtained for two integer values $t_1=5$ and $t_2=6$ while the correct decodable value must be $\hat{t}=5$ for the corresponding values of $x_1$ and $x_2$. This happens at high SNR range for which the maximization of the computation rate requires to align $\mathbf{a}$ to $\mathbf{h}$.  \\ 
\emph{- Impact of the constellation size:} the likelihood function depends on the constellation size and the values of $S_m$. When the size of the codebook increases, the set $\mathcal{A}_t$ over which the desired combination $t$ should be searched becomes large. Consequently, the width of $\varphi$ becomes large and the likelihood function is made flat. Thus, decoding the maximal value of $t$ becomes ambiguous. An example of this scenario is illustrated in Fig.\ref{flat3} obtained for   
 $S_m=10, \mathrm{SNR}=10\mathrm{dB}, x_1=-2, x_2=-4, \mathbf{h}=[1.4741~-0.2839]^{t}, \mathbf{a}=[-1~0]^{t}$. We can see that the likelihood function attains its maximum for $t=2$ and $t=3$ while the correctly decoded value is $\hat{t}=2$. This ambiguity leads to decoding errors.
\subsubsection{Diophantine Approximation}
The sum of Gaussian functions in the likelihood function makes the ML decoding hard to handle in practice. For this purpose, we use the result stating that for $t \in \mathbb{Z}$, $\varphi$ is maximized for $t$ which minimizes $\mid \tilde{y} - \gamma t + \beta k \mid$. Given this observation, we define a new optimization problem equivalent to (\ref{topt}) by:
\begin{equation}
\hat{t} = \argmin_{k \in \mathbb{Z}, t \in \mathcal{A}_t}  \mid \tilde{y} - \gamma t + \beta k \mid 
\end{equation}
Let $\beta^{'} = \frac{\beta}{\gamma}$ and $y^{'} = -\frac{\tilde{y}}{\gamma}$, then this minimization problem is equivalent to:
\begin{equation}\label{ida}
\hat{t} = \argmin_{k \in \mathbb{Z}, t \in \mathcal{A}_t}   \mid \beta^{'}k - t - y^{'} \mid 
\end{equation}
This problem corresponds to solving the \textit{Inhomogeneous Diophantine Approximation in the absolute sens} (IDA) \cite{Clarkson97}, $F(t,k)$, defined as, $F(t ,k)= \mid \beta^{'}k - t - y^{'} \mid$. It consists in finding the best rational approximation $\frac{t}{k}, k \in \mathbb{Z}$ of the real number $\beta^{'}$ assumed an additional real shift $y^{'}$. In our setting, the set of the diophantine approximations is determined by the limits imposed by the shaping boundaries $\mathcal{A}_t$. In literature, there exist simple and easy-to-implement algorithms to solve Diophantine Approximations of reals. The best known one is the \textit{Cassel's Algorithm} \cite{Cassel57}. In this work we adopt a modified version of this algorithm to take into consideration the shaping constraint and ensure that the resulting solution $(t, k)$ satisfies $t \in \mathcal{A}_t$.

\subsubsection{Simulation results}
We address now the performance evaluation of the conventional decoder and the proposed Inhomogenous Diophantine
Approximation (IDA) decoder. We consider the same settings analyzed previously involving two sources transmitting integer symbols $x_1$ and
$x_2$ drawn from the constellation set
$\mathcal{A}=[-S_m~S_m]$. We analyze the error probability on decoding $t$. For what concerns the conventional decoder, the receiver solves for the best network code vector $\mathbf{a}$ solution of the shortest vector problem, scales the channel output, then decodes to the nearest integer value. For the IDA, given the vector $\mathbf{a}$, the receiver implements first the \textit{Extended Euclid} algorithm to solve the Diophantine equation $a_1x_1+a_2x_2=g$, then uses the modified Cassel's algorithm to find the best inhomogeneous Diophantine approximation. In Fig.\ref{perf1ml} minimum distance decoding and IDA decoding are compared for $S_m=5$. Our results show that both decoding methods achieve same performance in low and moderate SNR values. The importance of the IDA method rises asymptotically, since for this case, the conventional decoder presents a floor in the error probability. In Fig.\ref{perf2ml}, we analyze the performance of the proposed IDA decoding for three values of the constellation bound, defined by $S_m=5,7,10$. This is to understand the impact of the constellation size on the diversity order. Fig.\ref{perf2ml} illustrates that for $S_m=5$ or less, the system has a diversity order equal to 1 for real symbols (which would correspond to a diversity order equal to 2 with complex-valued symbols). However, for higher constellation size, e.g., for $S_m=7$ and $S_m=10$, the diversity order is limited to $1/2$. This is because when the constellation range increases, the likelihood function becomes flat, which makes the error function $F(t,k)$ subject to the diophantine approximation flat. This result confirms our previous analysis on the impact of the constellation on the likelihood function.
\begin{figure}[h]
 \centering
 \includegraphics[height=9cm,width=11.5cm]{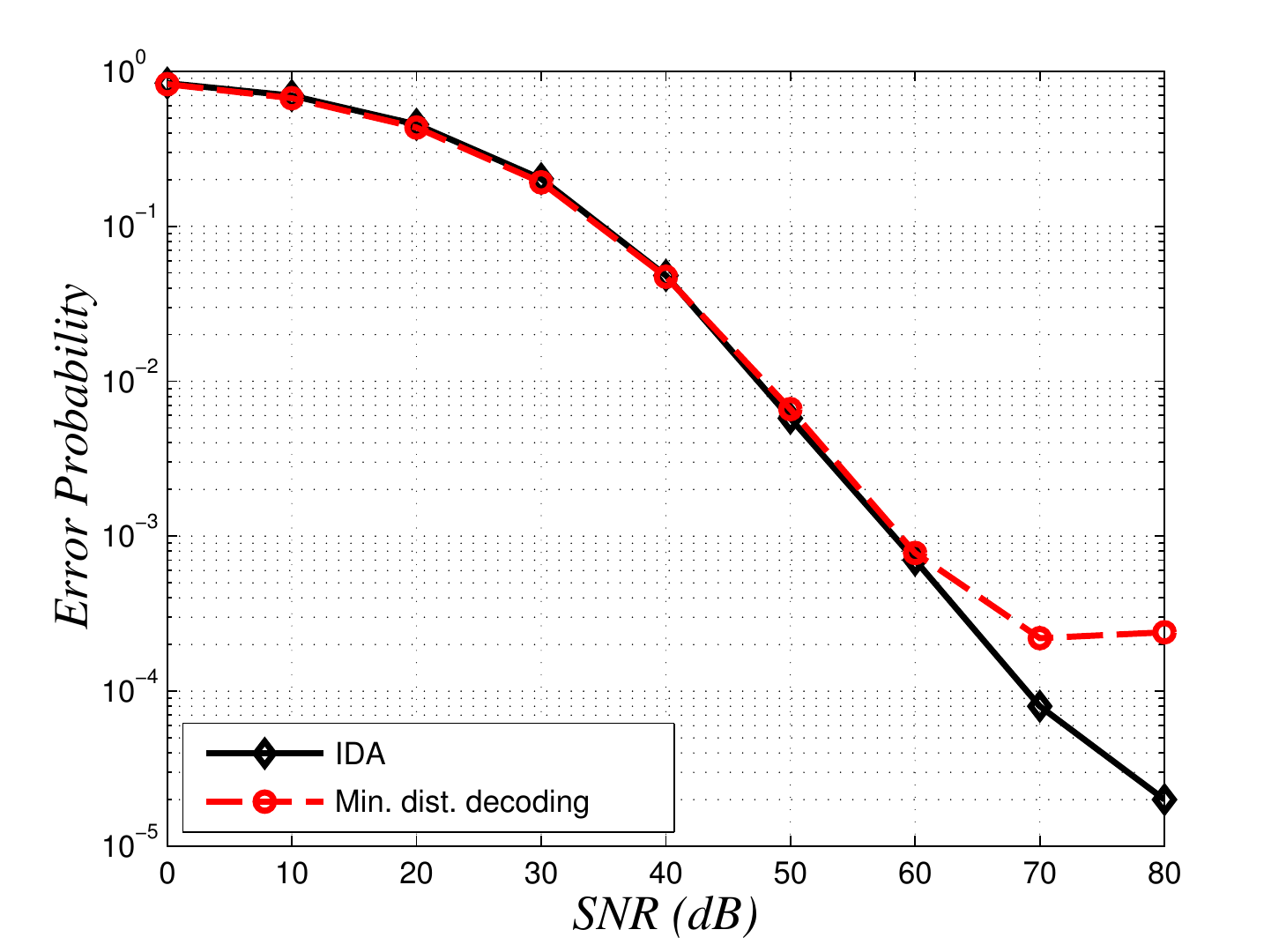}
 \caption{Error Probability for $S_m=5$.} \label{perf1ml}
 \end{figure}
  \begin{figure}[h]
   \centering
   \includegraphics[height=9cm,width=11.5cm]{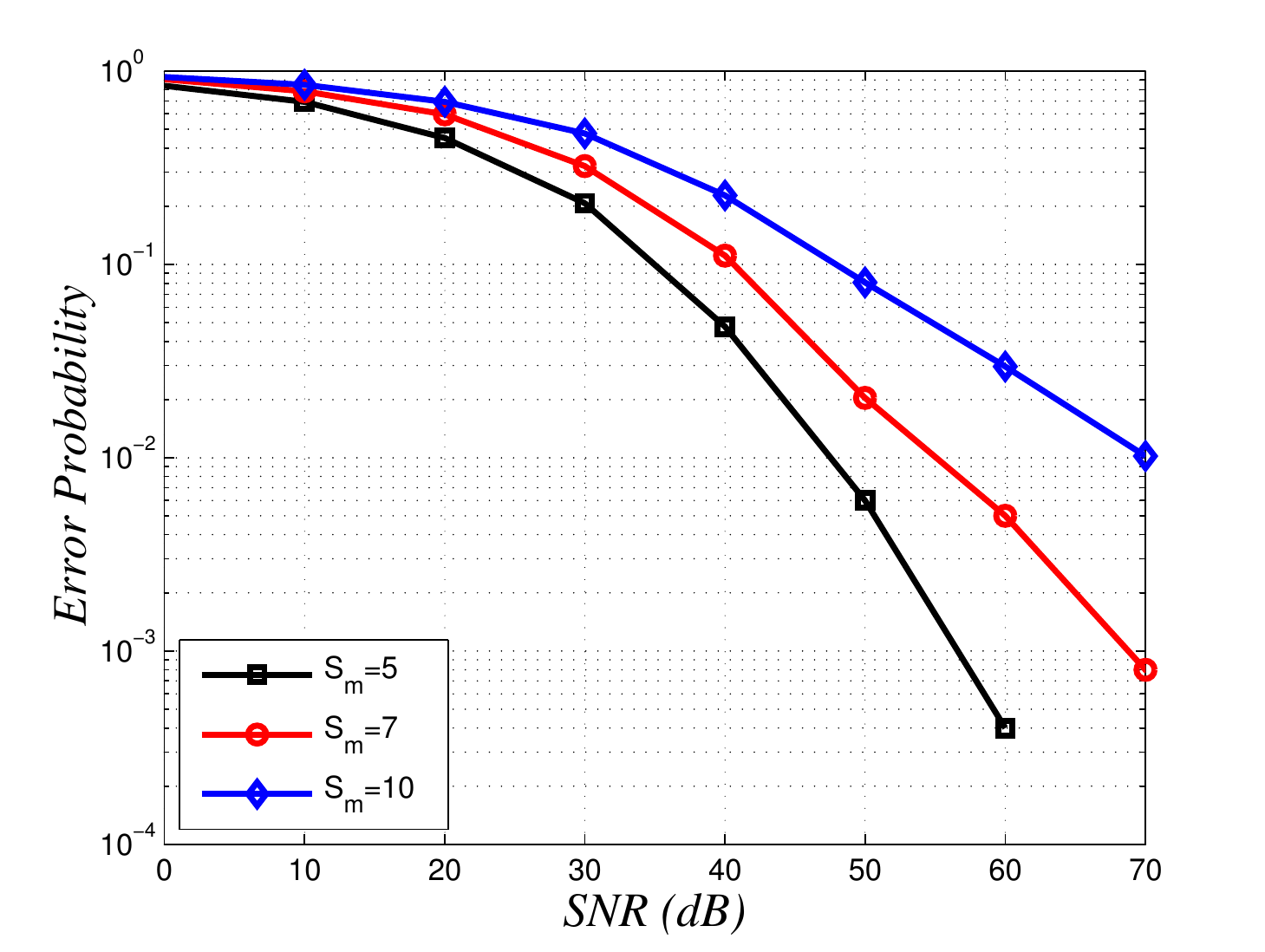}
   \caption{Error Probability using the Inhomogeneous Diophantine approximation.} \label{perf2ml}
   \end{figure}
\section{Efficient Decoders in Gaussian Channels}
\subsection{Problem Statement}
The system we are interested in within now is the real-valued Gaussian multiple access channel using real-valued nested lattice coding. The channel output is given by: $\mathbf{y} = \sum_{i=1}^{N} \mathbf{x}_i + \mathbf{z}$. The receiver aims to decode the noiseless sum $\lambda= \left[ \sum_{i=1}^{N} \mathbf{x}_i \right] \mathrm{mod}~\Lambda_{\mathrm{C}}$. 

Let $\Lambda_{\mathrm{s}}$ denote the \textit{sum codebook} which is the set of all $\lambda_{\mathrm{s}}=\sum_{i=1}^{N} \mathbf{x}_i$. Given the linear structure of the coding lattice, $\Lambda_{\mathrm{s}}$ will be a subset of the fine lattice $\Lambda_{\mathrm{F}}$ restricted to a \textit{sum shaping region} $\mathcal{S}_{\mathrm{s}}$ such that all sum codewords $\lambda_{\mathrm{s}}$ fall within this region. In addition, given that $\Lambda_{\mathrm{s}}$ is obtained through a superposition of the originally transmitted codewords, its distribution is \textbf{no longer uniform}. 

Using the conventional CF decoder, the receiver decodes  $\lambda_{\mathrm{s}}=\sum_{i=1}^{N} \mathbf{x}_i$ using an MMSE scaling followed by minimum distance decoding to the nearest point in the fine lattice. Using this method, there are three fundamental limitations: $i)$ the effective noise $\mathbf{z}_{\mathrm{eq}}=\sum_{i=1}^{N} \left(1-\alpha\right)\mathbf{x}_i+\alpha\mathbf{z}$ is not Gaussian, $ii)$ the shaping constraint is disregarded, and $iii)$ the non uniform distribution of the sum codebook $\Lambda_{\mathrm{s}}$ is not taken into account. A main contribution of this work is the analysis in the following of the optimal MAP decoding approach that takes into consideration the above mentionned drawbacks of the conventional CF decoder. To the best of our knowledge, this is the first investigation of the MAP decoder for the CF protocol. We will be interested in decoding $\lambda_{\mathrm{s}}=\sum_{i=1}^{N} \mathbf{x}_i$ given that modulo-lattice operation is done separately and does not impact the decoding error. We will evaluate thereforre the error probability at the receiver as $\mathrm{P}_\mathrm{e}= \mathrm{Pr} \left( \hat{\lambda}_s \neq \lambda_s \right)$.
\subsection{MAP Decoder: Error Probability and Lattice Design Criterion}
Under the non-uniform distribution of the sum codebook, the optimal decoder that minimizes the probability of decoding error at the receiver is the \textit{maximum a posteriori} decoder given according to the following:
\begin{align}\label{map}
\hat{\lambda}_{\mathrm{map}} &= \argmax_{\lambda_{\mathrm{s}} \in \Lambda_{\mathrm{s}}}  p\left(\mathbf{\lambda_{\mathrm{s}}} | \mathbf{y} \right)= \argmax_{\lambda_{\mathrm{s}} \in \Lambda_{\mathrm{s}}} \left\lbrace p(\lambda_{\mathrm{s}}) \frac{1}{(\sigma\sqrt{2\pi})^{n}} \exp{\left(- \frac{\parallel \mathbf{y}-\lambda_{\mathrm{s}} \parallel^{2}}{2\sigma^{2}}\right)} \right\rbrace \notag \\
		      &= \argmin_{\lambda_{\mathrm{s}} \in \Lambda_{\mathrm{s}}} \left\lbrace -\ln{\left(p(\lambda_{\mathrm{s}})\right)} + \frac{\parallel \mathbf{y}-\lambda_{\mathrm{s}} \parallel^{2}}{2\sigma^{2}}\right\rbrace
\end{align}
Notice that the MAP decoder does not involve a scaling step like the conventional decoder keeping the channel noise Gaussian. A first contribution in this context consists in deriving in the following theorem a union bound estimate on the decoding error probability.
\begin{theorem}
Consider a nested lattice design $\Lambda=(\Lambda_{\mathrm{F}},\Lambda_{\mathrm{C}})$ and a receiver computing a noiseless sum of $N$ source codewords in a Gaussian MAC using the optimal maximum a posteriori decoder. Then the union bound estimate of the probability of decoding error is:
\begin{equation}\label{proba2}
\mathrm{P}_{\mathrm{e}}  \leq \frac{1}{2} \sum_{\lambda_\mathrm{s} \in \Lambda_{\mathrm{s}}}\hspace{0.25cm} \sum_{\hat{\lambda}_{\mathrm{s}} \in \Lambda_{\mathrm{s}} \setminus \lambda_{\mathrm{s}}} p(\lambda_{\mathrm{s}}) \mathrm{erfc}\left( \sqrt{A} + \frac{B}{\sqrt{A}} \right)
\end{equation}
where $A=\frac{d_{\mathrm{min}}^{2}}{8\sigma^{2}}$, $B=\frac{1}{4}\ln{\left(\frac{p(\lambda_{\mathrm{s}})}{p(\hat{\lambda}_{\mathrm{s}})}\right)}$ and $d_{\mathrm{min}}$ denotes the minimum distance of the fine lattice $\Lambda_{\mathrm{F}}$.
\end{theorem}
\begin{proof}
The proof of our theorem is based on the pairwise error probability defined as the probability that the sum codeword 
$\lambda_{\mathrm{s}}$ has a larger MAP decoding metric in (\ref{map}) than $\hat{\lambda}_{\mathrm{s}}$ given that 
$\lambda_{\mathrm{s}}$ is transmitted. Its expression is formulated as follows
\begin{align}
\mathrm{Pr}(\lambda_{\mathrm{s}} \longrightarrow \hat{\lambda}_{\mathrm{s}}) &=  \footnotesize \mathrm{Pr} \left(  -\ln{\left(p(\hat{\lambda}_{\mathrm{s}})\right)} + \frac{\parallel \mathbf{y}-\hat{\lambda}_{\mathrm{s}} \parallel^{2}}{2\sigma^{2}} < -\ln{\left(p(\lambda_{\mathrm{s}})\right)} + \frac{\parallel \mathbf{y}-\lambda_{\mathrm{s}} \parallel^{2}}{2\sigma^{2}} \right) \notag \\
&= \mathrm{Pr} \left(  \ln{\left(\frac{p(\lambda_{\mathrm{s}})}{p(\hat{\lambda}_{\mathrm{s}})}\right)} +  \frac{\parallel \mathbf{y}-\hat{\lambda}_{\mathrm{s}} \parallel^{2}}{2\sigma^{2}}- \frac{\parallel \mathbf{y}-\lambda_{\mathrm{s}} \parallel^{2}}{2\sigma^{2}} < 0  \right) \notag \\
&= \mathrm{Pr} \left(2\sigma^{2}\ln{\left(\frac{p(\lambda_{\mathrm{s}})}{p(\hat{\lambda}_{\mathrm{s}})}\right)}+ \parallel \lambda_{\mathrm{s}}-\hat{\lambda}_{\mathrm{s}} \parallel^{2} + 2 \left\langle \lambda_{\mathrm{s}} -\hat{\lambda}_{\mathrm{s}},\mathbf{z} \right\rangle < 0  \right) \notag \\
&= \mathrm{Pr} \left( G < 0 \right)= \mathrm{Q}\left( \frac{\mu_{G}}{\sigma_{G}} \right)= \mathrm{Q}\left( \frac{\parallel \lambda_{\mathrm{s}} - \hat{\lambda}_{\mathrm{s}} \parallel}{2\sigma} + \frac{\sigma}{\parallel \lambda_{\mathrm{s}} - \hat{\lambda}_{\mathrm{s}} \parallel}\ln{\left(\frac{p(\lambda_{\mathrm{s}})}{p(\hat{\lambda}_{\mathrm{s}})}\right)} \right) \notag
\end{align}
Where $\mathrm{Q}(.)$ denotes the $\mathrm{Q}$ function and it is easy to prove that $$G=2\sigma^{2}\ln{\left(\frac{p(\lambda_{\mathrm{s}})}{p(\hat{\lambda}_{\mathrm{s}})}\right)}+ \parallel \lambda_{\mathrm{s}}-\hat{\lambda}_{\mathrm{s}} \parallel^{2} + 2<\lambda_{\mathrm{s}} -\hat{\lambda}_{\mathrm{s}},\mathbf{z}>$$
 is a random Gaussian variable of mean $\mu_{\mathrm{G}}$ and variance $\sigma_{\mathrm{G}}^{2}$ given by:
\begin{align}
\mu_{G} = \parallel \lambda_{\mathrm{s}} - \hat{\lambda}_{\mathrm{s}} \parallel^{2} + 4 \sigma^{2} \ln{\left(\frac{p(\lambda_{\mathrm{s}})}{p(\hat{\lambda}_{\mathrm{s}})}\right)} ~,~
\sigma_{G}^{2} = 4 \sigma^{2} \parallel \lambda_{\mathrm{s}} - \hat{\lambda}_{\mathrm{s}} \parallel ^{2}
\end{align}
Using the union bound, we get,
\begin{align}
\mathrm{P}_{\mathrm{e}} & \leq \sum_{\lambda_{\mathrm{s}} \in \Lambda_{\mathrm{s}}} p(\lambda_{\mathrm{s}}) \sum_{\hat{\lambda}_{\mathrm{s}} \in \Lambda_{\mathrm{s}} \setminus \lambda_{\mathrm{s}}} \mathrm{Pr}(\lambda_{\mathrm{s}} \longrightarrow \hat{\lambda}_{\mathrm{s}}) \notag \\
& \leq \sum_{\lambda_{\mathrm{s}} \in \Lambda_{\mathrm{s}}} \hspace{0.25cm} \sum_{\hat{\lambda}_{\mathrm{s}} \in \Lambda_{\mathrm{s}} \setminus \lambda_{\mathrm{s}}} p(\lambda_{\mathrm{s}}) \mathrm{Q}\left( \frac{\parallel \lambda_{\mathrm{s}} - \hat{\lambda}_{\mathrm{s}} \parallel}{2\sigma} + \frac{\sigma}{\parallel \lambda_{\mathrm{s}} - \hat{\lambda}_{\mathrm{s}} \parallel}\ln{\left(\frac{p(\lambda_{\mathrm{s}})}{p(\hat{\lambda}_{\mathrm{s}})}\right)} \right) \notag
\end{align}
 We can therefore, using the relation $\mathrm{Q}(x)=\frac{1}{2}\mathrm{erfc}(\frac{x}{\sqrt{2}})$, write:
\begin{equation}\label{proba1}
\mathrm{P}_{\mathrm{e}}  \leq \frac{1}{2} \sum_{\lambda_{\mathrm{s}} \in \Lambda_{\mathrm{s}}} \hspace{0.25cm} \sum_{\hat{\lambda}_{\mathrm{s}} \in \Lambda_{\mathrm{s}} \setminus \lambda_{\mathrm{s}}} p(\lambda_{\mathrm{s}}) \mathrm{erfc}\left( \frac{\parallel \lambda_{\mathrm{s}} - \hat{\lambda}_{\mathrm{s}} \parallel}{2\sqrt{2}\sigma} + \frac{\sigma}{\sqrt{2}\parallel \lambda_{\mathrm{s}} - \hat{\lambda}_{\mathrm{s}} \parallel}\ln{\left(\frac{p(\lambda_{\mathrm{s}})}{p(\hat{\lambda}_{\mathrm{s}})}\right)} \right) \notag
\end{equation}
The last step to prove our theorem is based on two facts: $i)$ for all $\lambda_{\mathrm{s}}, \hat{\lambda}_{\mathrm{s}} \in \Lambda_{\mathrm{s}}$, $\parallel \lambda_{\mathrm{s}} - \hat{\lambda}_{\mathrm{s}} \parallel \geq d_{\mathrm{min}}$. This inequality results from the linear structure and the geometrical symmetric properties of the fine lattice $\Lambda_{\mathrm{F}}$, and $ii)$ $\mathrm{erfc}(x+\frac{\alpha}{x}),~\alpha \in \mathbb{R}$ is a decreasing function with respect to $x$ \cite{Behnamfar03}.
The proof follows then by considering $A$ and $B$ as defined above. 
\end{proof}
 Given the derived upper bound, we propose a lattice design criterion as follows.
 \begin{proposition}
 Minimization of the error probability under MAP decoding requires to design nested lattices  $\Lambda=\left(\Lambda_{\mathrm{F}},\Lambda_{\mathrm{C}} \right)$ such that the minimum distance of the Fine lattice is maximized.
 \end{proposition}
 \begin{proof}
 The upper bound on the error probability is a strictly decreasing function of $A$ \cite{Behnamfar03}, thus a decreasing function of the minimum distance of the lattice $\Lambda_{\mathrm{F}}$. Then in order to make the error probability small, the coding lattice $\Lambda_{\mathrm{F}}$ has to have a large minimum distance $d_{\mathrm{min}}$.
 \end{proof}
 The construction of such good codes is out of the scope of this work. Even though, we point out that for lattices built using Construction A \cite{ErezZamir} over linear codes, this criterion requires to design linear codes with minimum euclidean weights.
\subsection{Practical MAP Decoding Algorithms} 
We aim in this section to develop practical decoding algorithms that allow to reliably find the optimal MAP estimate of the optimization problem in (\ref{map}). For this purpose, we study first the statistical distribution of the sum codewords.

The original codewords are drawn uniformally and independently from the nested lattice code, they are modeled by uniform random variables of zero-mean ($\mu_{\mathbf{x}}=0$) and variance $\sigma_{\mathbf{x}}^{2}=\frac{1}{n} \mathbb{E}\left( \parallel \mathbf{x}_i \parallel ^{2} \right) \leq P$ for $i=1,...,N$.
 Consider now the sum codewords $\lambda_{\mathrm{s}}=\sum_{i=1}^{N} \mathbf{x}_i$ obtained through the superposition of the vectors sent by the sources. Given the uniform distribution of the original codewords, The \textit{Central Limit Theorem} states that $\lambda_{\mathrm{s}}$ is a random variable of mean $\mu_{\mathrm{s}}=N\mu_{\mathbf{x}}=0$ and variance $\sigma_{\mathrm{s}}^{2}=N\sigma_{\mathbf{x}}^{2}$. Particularly, for increasing number of sources $N$, the sum codewords converge to the normal distribution $\mathcal{N}\left(\mu_{\mathrm{s}},\sigma_{\mathrm{s}}^{2}\mathbf{I}_n \right)$. In order to be able to use this result to approximate the vectors $\lambda_{\mathrm{s}}$ by random Gaussian variables, we need in addition to take into consideration the fact that the sum codewords are discrete and correspond to lattice points. For this purpose we introduce the lattice Gaussian distributions. This tool arises in several problems in coding theory \cite{Forney00}, mathematics \cite{Banaszczyk93} and cryptography \cite{Micciancio04}.
 
  Let $f_{\sigma_{\mathrm{s}}}(\mathbf{x})$ denote the Gaussian distribution of variance $\sigma_{\mathrm{s}}^{2}$ centered at the zero vector such that for $\sigma_{\mathrm{s}} > 0$ and all $\mathbf{x} \in \mathbb{R}^{n}$:
   $$f_{\sigma_{\mathrm{s}}}(\mathbf{x}) = \frac{1}{\left(\sqrt{2\pi}\sigma_{\mathrm{s}} \right)^{n}} e^{-\frac{\parallel \mathbf{x}\parallel^{2} }{2\sigma_{\mathrm{s}}^{2}}}$$
    Consider also the $\Lambda_{\mathrm{F}}-$periodic function defined by:
$$f_{\sigma_{\mathrm{s}}}(\Lambda_{\mathrm{F}}) = \sum_{\lambda_{\mathrm{s}} \in \Lambda_{\mathrm{F}}} f_{\sigma_{\mathrm{s}}}(\lambda_{\mathrm{s}}) = \frac{1}{\left(\sqrt{2\pi}\sigma_{\mathrm{s}} \right)^{n}} \sum_{\lambda_{\mathrm{s}} \in \Lambda_{\mathrm{F}}} e^{-\frac{\parallel \lambda_{\mathrm{s}} \parallel^{2} }{2\sigma_{\mathrm{s}}^{2}}}$$
    Then the sum codewords can be modeled by the discrete Gaussian distributions over $\Lambda_{\mathrm{F}}$ centered at the zero vector according to: $p(\lambda_{\mathrm{s}})= \frac{f_{\sigma_{\mathrm{s}}}(\lambda_{\mathrm{s}})}{f_{\sigma_{\mathrm{s}}}(\Lambda_{\mathrm{F}})}$. 
    \begin{figure}[htp]
        \centering
        \subfigure[\scriptsize Histogram for N=2.]{\includegraphics[height=7cm,width=6.5cm]
        {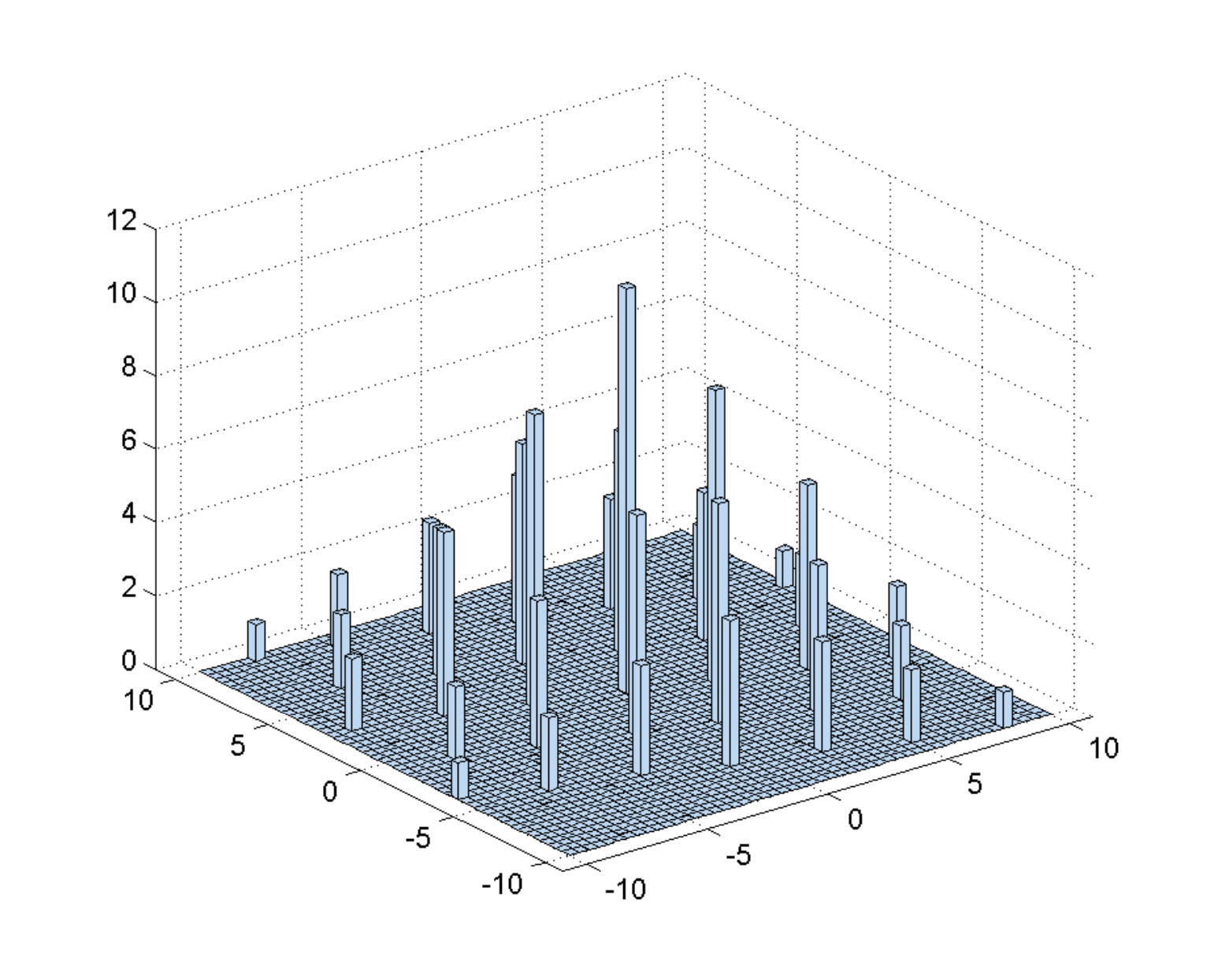}} \
        \subfigure[\scriptsize Histogram for N=5.]{\includegraphics[height=7cm,width=6.5cm]{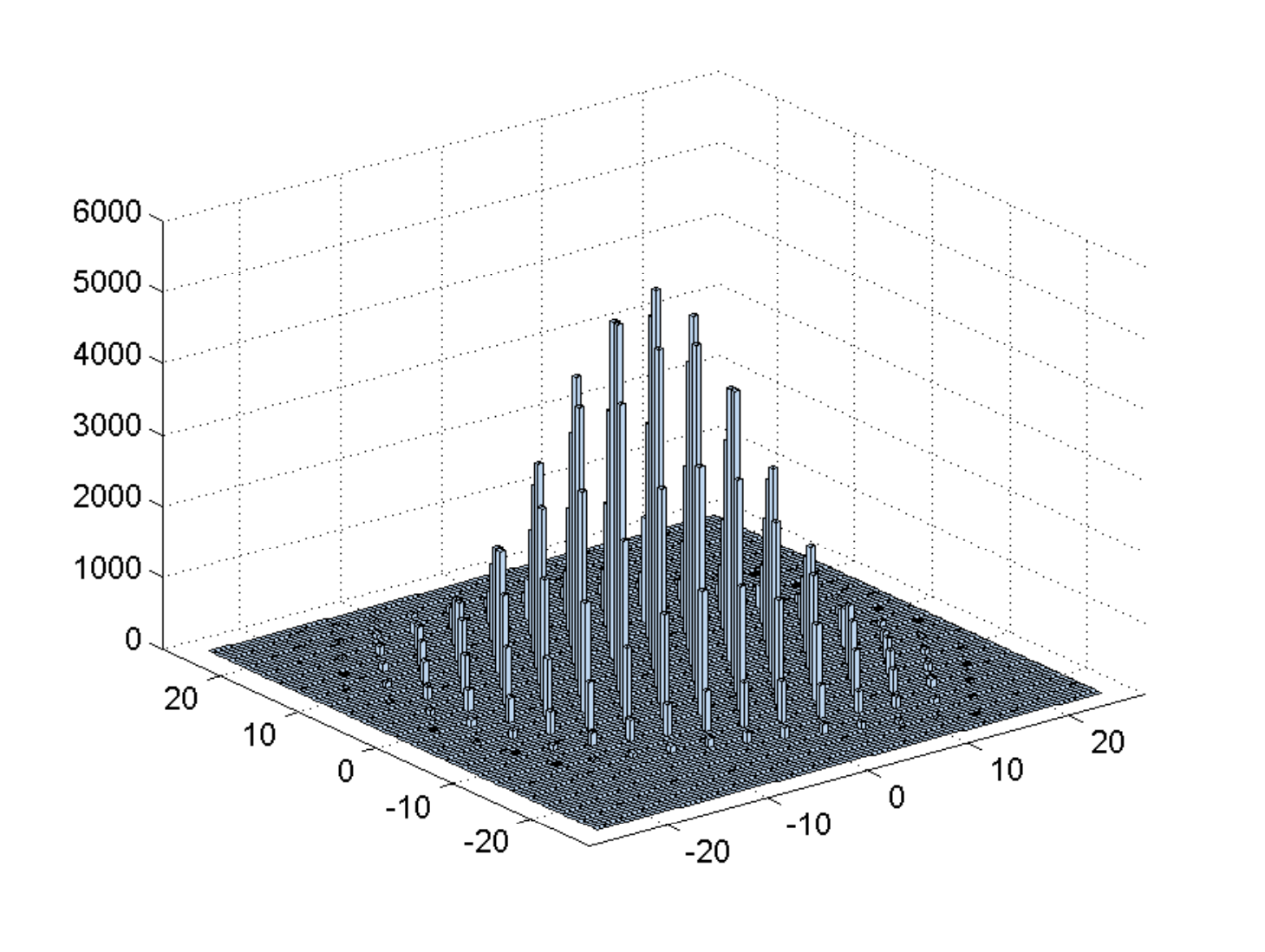}}
        \caption{Histogram of the codebook induced by the sum of codewords.}\label{nested}
        \end{figure}\\   
    We illustrate in Fig.\ref{nested} two examples of the statistical distribution of the sum codebook resulting from the superposition of 2-dimensional lattice codewords for the cases of $N=2$ and $N=5$ considering a Fine lattice $\Lambda_{\mathrm{F}}$ of a generator matrix $\mathbf{M}=\left[ \begin{array}{cc}
         2 & 3 \\ 3 & -1 \end{array} \right]$ and the coarse lattice $\Lambda_{\mathrm{C}}=11\mathbb{Z}^{2} (k=1,p=11)$. These examples show that the lattice Gaussian distribution fits our settings. As a proof of concept, we will show by numerical results that this model is well justified in the context of lattice network coding even for low number of sources.
            
Using this Gaussian distribution, the MAP decoding rule in (\ref{map}) is equivalent to:  
\begin{equation}\label{map1}
\hat{\lambda}_{\mathrm{map}}= \argmin_{\lambda_{\mathrm{s}} \in \Lambda_{\mathrm{s}}} \left\lbrace \ln{\left(f_{\sigma_{\mathrm{s}}}(\Lambda_{\mathrm{F}}) \right)}+ n\ln{(\sigma_{\mathrm{s}}\sqrt{2\pi})} +\frac{\Vert \lambda_{\mathrm{s}} \Vert ^{2}}{2\sigma_{\mathrm{s}}^{2}}+\frac{\Vert \mathbf{y}-\lambda_{\mathrm{s}} \Vert ^{2}}{2\sigma^{2}}  \right\rbrace \notag 
\end{equation}
The first and second terms in this metric are independent of the variable $\lambda_{\mathrm{s}}$, they can be disregarded in the optimization over $\lambda_{\mathrm{s}}$. Then we obtain our novel decoding metric given by: 
\begin{align}\label{mapeq}
\hat{\lambda}_{\mathrm{map}} &= \argmin_{\lambda_{\mathrm{s}} \in \Lambda_{\mathrm{s}}} \left\lbrace \parallel \mathbf{y}- \lambda_{\mathrm{s}} \parallel ^{2} + \beta^{2} \parallel \lambda_{\mathrm{s}} \parallel ^{2}  \right\rbrace   
\end{align}
where $\beta=\frac{\sigma}{\sigma_{\mathrm{s}}}$. Using this new metric, we show in Proposition.\ref{mapmeth1} that MAP decoding reduces to solve for a closest vector problem. 
\begin{proposition}\label{mapmeth1}
The MAP decoding metric in (\ref{mapeq}) is equivalent to 
find the closest vector in the lattice $\Lambda_{\mathrm{aug}}$ of generator matrix $\mathbf{M}_{\mathrm{aug}}=\left[ \mathbf{M}~\beta \mathbf{M}\right]^{t} \in \mathbb{R}^{2n \times n}$ to the vector $\mathbf{y}_{\mathrm{aug}}=\left[\mathbf{y}~  \mathbf{0}_{n}\right]^{t}$ according to the following metric:
\begin{align}
\hat{\lambda}_{\mathrm{map}} &= \argmin_{
\mathbf{x}_{\mathrm{aug}} \in \Lambda_{\mathrm{aug}}/\mathbf{x}_{\mathrm{aug}}=\mathbf{M}_{\mathrm{aug}} \lambda_{\mathrm{s}}} \parallel \mathbf{y}_{\mathrm{aug}} - \mathbf{x}_{\mathrm{aug}} \parallel^{2}  
\end{align}
\end{proposition}
\begin{proof}
The decoding metric in (\ref{mapeq}) can be written as:
\begin{align}\label{mapeq1}
\hat{\lambda}_{\mathrm{map}} &= \argmin_{\lambda_{\mathrm{s}} \in \Lambda_{\mathrm{s}}} \left\lbrace \left\|  \left[ \begin{array}{c}
\mathbf{y} \\ \mathbf{0}_{n}
\end{array}\right] -  \left[ \begin{array}{c}
\lambda_{\mathrm{s}} \\ \beta \lambda_{\mathrm{s}}
\end{array}\right]  \right\| ^{2} \right\rbrace= \argmin_{\lambda_{\mathrm{s}} \in \Lambda_{\mathrm{s}}} \parallel \mathbf{y}_{\mathrm{aug}} - \mathbf{I}_{\mathrm{aug}} \lambda_{\mathrm{s}} \parallel^{2}  
\end{align}
where $\mathbf{I}_{\mathrm{aug}}=\left[ \mathbf{I}_{n}~\beta \mathbf{I}_{n} \right]^{t} \in \mathbb{R}^{2n \times n}$ is a full rank matrix. On the other hand, given that the sum codewords belong to the fine lattice according to the shaping region $\mathcal{S}_{\mathrm{s}}$, any codeword $\lambda_{\mathrm{s}}$ can be written in the form $\lambda_{\mathrm{s}}=\mathbf{M}\mathbf{u}$ where $\mathbf{u} \in \mathcal{A}_{\mathrm{s}} \subset \mathbb{Z}^{n}$ and $\mathcal{A}_{\mathrm{s}}$ translates the shaping constraint imposed by $\mathcal{S}_{\mathrm{s}}$ and can be deduced from the shaping boundaries limited by the transmission power constraint $P$. Consequently the optimization problem in (\ref{mapeq1}) is equivalent to solving
\begin{align}\label{map4}
\hat{\lambda}_{\mathrm{map}} &= \argmin_{\mathbf{u} \in \mathcal{A}_{\mathrm{s}} / \lambda_{\mathrm{s}}=\mathbf{M}\mathbf{u} }   \parallel \mathbf{y}_{\mathrm{aug}} - \mathbf{I}_{\mathrm{aug}} \mathbf{M} \mathbf{u} \parallel^{2}= \argmin_{\mathbf{u} \in \mathcal{A}_{\mathrm{s}} / \lambda_{\mathrm{s}}=\mathbf{M}\mathbf{u} }  \parallel \mathbf{y}_{\mathrm{aug}} - \mathbf{M}_{\mathrm{aug}} \mathbf{u} \parallel^{2} 
\end{align}
$\mathbf{M}_{\mathrm{aug}}$ is a full rank matrix and $\mathbf{u}$ is an integer vector, then solving (\ref{map4}) consists in finding the closest vector $\mathbf{x}_{\mathrm{aug}}=\mathbf{M}_{\mathrm{aug}}\mathbf{u}$ to $\mathbf{y}_{\mathrm{aug}}$ in the $n-$dimensional lattice $\Lambda_{\mathrm{aug}}$ of a generated matrix $\mathbf{M}_{\mathrm{aug}}$. After finding the optimal integer vector $\mathbf{u}_{\mathrm{opt}}$ that minimizes the metric in (\ref{map4}), the optimal MAP estimate is deduced by $\hat{\lambda}_{\mathrm{map}}=\mathbf{M}\mathbf{u}_{\mathrm{opt}}$. 
\end{proof}
In our implementation, we use a modified version of the sphere decoder to solve this closest vector problem taking into account the shaping constraint.
\begin{remark}\label{rmqmap}
The MAP decoding metric in (\ref{mapeq}) involves two terms each one of them is given by an Euclidean distance. When the first term is dominant, which is the case when $\beta^{2}=\frac{\sigma^{2}}{\sigma_{\mathrm{s}}^{2}}= \frac{\sigma^{2}}{N\sigma_{\mathbf{x}}^{2}} \ll 1$, the MAP decoding rule reduces to ML decoding (which is equivalent to minimum distance decoding in this case since we don't perform a scaling step). Given that $\sigma_{\mathbf{x}}^{2}$ depends on the power constraint $P$, we deduce that this case of figure is likely to happen either at high Signal-to-Noise Ratio or when $N\sigma_{\mathbf{x}}^{2}$ is sufficiently higher than the noise variance $\sigma^{2}$. We expect then that the MAP decoding and the conventional decoder achieve similar performance at high SNR range. Adversely, at the low and moderate SNR regime and when the product $N\sigma_{\mathbf{x}}^{2}$ is small, the second term in the decoding metric applies an incremental constraint that considers the non-uniform distribution of the sum codewords in $\Lambda_{\mathrm{s}}$ which is not taken into account under the conventional decoder. In this case, we expect that the MAP decoder outperforms the minimum distance decoding-based one.
\end{remark}
We provide in the following proposition an equivalent formulation of the MAP decoding metric.
\begin{proposition}\label{mapmeth2}The MAP decoding metric in (\ref{mapeq}) is equivalent to MMSE-GDFE preprocessed minimum Euclidean distance decoding according to the metric:
\begin{align}\label{mmsegdfe}
\hat{\lambda}_{\mathrm{map}} &= \argmin_{\lambda_{\mathrm{s}} \in \Lambda_{\mathrm{s}}}  \parallel \mathbf{F} \mathbf{y}-\mathbf{B} \lambda_{\mathrm{s}} \parallel ^{2}   
\end{align}
$\mathbf{F} \in \mathbb{R}^{n \times n}$ and $\mathbf{B} \in \mathbb{R}^{n \times n}$ denote respectively the forward and backward filters of the MMSE-GDFE preprocessing for the channel $\mathbf{y}=\lambda_{\mathrm{s}}+\mathbf{z}$ such that $\mathbf{B}^{t}\mathbf{B}=\left(1+\beta^{2}\right)\mathbf{I}_n$ and $\mathbf{F}^{t}\mathbf{B}=\mathbf{I}_n$.
\end{proposition}
\begin{proof}
Let $N(\lambda_{\mathrm{s}})$ denote the metric we aim to minimize in (\ref{mapeq}), we have the following:
\begin{align}
N(\lambda_{\mathrm{s}}) &= \parallel \mathbf{y}- \lambda_{\mathrm{s}} \parallel ^{2} + \beta^{2} \parallel \lambda_{\mathrm{s}} \parallel ^{2}= \mathbf{y}^{t}\mathbf{y}-2\mathbf{y}^{t}\lambda_{\mathrm{s}}+ \lambda_{\mathrm{s}}^{t}\lambda_{\mathrm{s}}+\beta^{2}\lambda_{\mathrm{s}}^{t}\lambda_{\mathrm{s}} \notag \\
&= \left(1+\beta^{2} \right)\lambda_{\mathrm{s}}^{t}\lambda_{\mathrm{s}}+\mathbf{y}^{t}\mathbf{y}-2\mathbf{y}^{t}\lambda_{\mathrm{s}}= \lambda_{\mathrm{s}}^{t}\mathbf{B}^{t}\mathbf{B}\lambda_{\mathrm{s}} + \mathbf{y}^{t}\mathbf{y}-2\mathbf{y}^{t}\mathbf{F}^{t}\mathbf{B}\lambda_{\mathrm{s}} \notag \\
&= \underbrace{\lambda_{\mathrm{s}}^{t}\mathbf{B}^{t}\mathbf{B}\lambda_{\mathrm{s}}+ \mathbf{y}^{t}\mathbf{F}^{t}\mathbf{F}\mathbf{y}-2\mathbf{y}^{t}\mathbf{F}^{t}\mathbf{B}\lambda_{\mathrm{s}}}_{\parallel \mathbf{F} \mathbf{y}-\mathbf{B} \lambda_{\mathrm{s}} \parallel ^{2}} + \underbrace{\mathbf{y}^{t} \left( \mathbf{I}_n - \mathbf{F}^{t}\mathbf{F} \right) \mathbf{y}}_{\Gamma(\mathbf{y})}
\end{align}
where $\mathbf{F} \in \mathbb{R}^{n \times n}$ and $\mathbf{B} \in \mathbb{R}^{n \times n}$ are chosen such that: $\mathbf{B}^{t}\mathbf{B}=\left(1+\beta^{2}\right)\mathbf{I}_n$ and $\mathbf{F}^{t}\mathbf{B}=\mathbf{I}_n$. Given that $\Gamma(\mathbf{y}) > 0$ and independent of $\lambda_{\mathrm{s}}$, minimization of $N(\lambda_{\mathrm{s}})$ is equivalent to minimize $\parallel \mathbf{F} \mathbf{y}-\mathbf{B} \lambda_{\mathrm{s}} \parallel ^{2}$. The last piece to our proof is to show that the matrices $\mathbf{F}$ and $\mathbf{B}$ correspond to the filters of the MMSE-GDFE preprocessing in the system $\mathbf{y}=\lambda_{\mathrm{s}}+\mathbf{z}$ of input $\lambda_{\mathrm{s}}$ and AWGN $\mathbf{z}$. This proof is provided in Appendix \ref{appendix:mmsegdfe}.
\end{proof}
In order to find the MAP estimate according to the
decoding metric in (\ref{mmsegdfe}), the receiver first performs MMSE-GDFE preprocessing, then performs minimum Euclidean distance decoding to find the nearest point to $\mathbf{F}\mathbf{y}$ in the lattice of generator matrix $\mathbf{B}\mathbf{M}$ according to the shaping constraint imposed by the subset $\Lambda_{\mathrm{s}}$. 

\subsection{Numerical Results}
We evaluate in this subsection the performance of the conventional decoder (based on MMSE scaling and minimum distance decoding) and the proposed MAP decoding algorithm implementing a modified sphere decoder. In addition, in order to validate the Gaussianity law assumption we considered to derive our MAP decoding metric, we include a naive exhaustive search to solve (\ref{map}). Using this approach, no assumptions on the sum codebook distribution is considered. The receiver, given the number of sources and the original codebook associated to the nested lattice $\Lambda$, derives the statistics of the sum codebook to compute the corresponding values of $p(\lambda_{\mathrm{s}})$ for all codewords $\lambda_{\mathrm{s}} \in \Lambda_{\mathrm{s}}$, then, it exhaustively seeks the codeword which maximizes the decoding metric in (\ref{map}). We study in our analysis two lattice examples as described below.\\
 \begin{figure}[h]
   \centering
   \includegraphics[height=9cm,width=11.5cm]{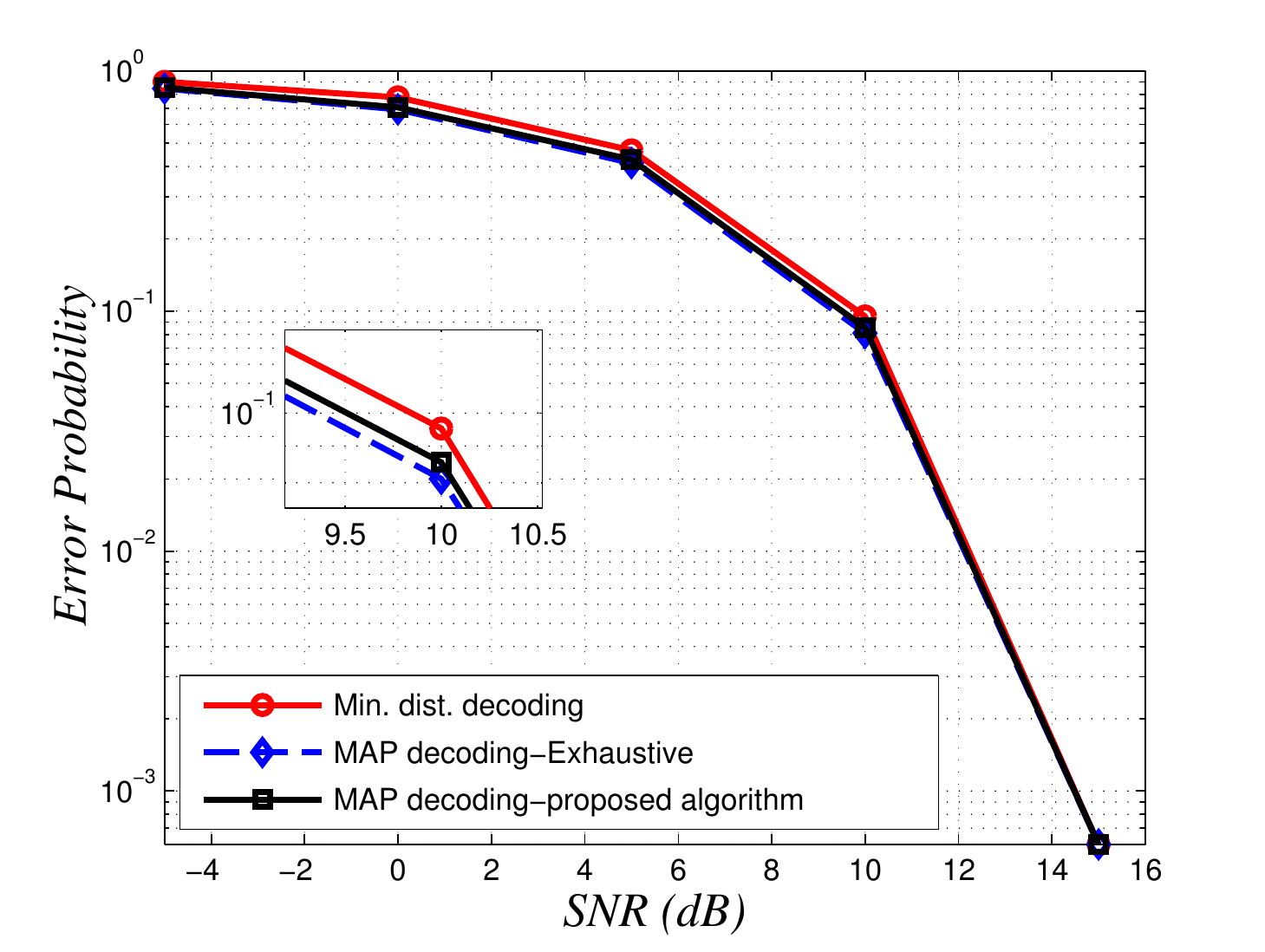}
   \caption{ Error performance for the case $n=2,N=2,P=21$.}\label{perf}
  \end{figure}
\emph{Example 1: 2-Dimensional lattice $(n=2)$} for this example we consider the same nested lattice code used to get the statistical distributions plotted in Fig.\ref{nested} for $N=2$ and $N=5$. The shaping constraint in this case is given by $P=\sigma_{\mathbf{x}}^{2}=6.5$. Given the number of sources and the power constraint imposed by the coarse lattice, we calculate for each case the bounds requirements to be considered in the decoding process. Numerical results concerning the case $N=2$, depicted in Fig.\ref{perf}, show that our proposed algorithm achieves almost identical performance as the exhaustive search, which confirms the effectiveness of our metric as well as the validity of the Gaussianity law assumption considered to model the sum-codewords even for the case of low number of sources $N=2$. Moreover, plotted curves show that the MAP decoder outperforms the conventional minimum distance decoding (Min. dist. decoding). The gain for this case is limited to $0.5$dB for an error probability equal to $10^{-1}$. 
Results for the case of $N=5$ plotted in Fig.\ref{perf11} confirm the previous findings and show that the performance gap between the MAP and the Minimum distance decoder is also not high. Common to these two settings is the high value of $N\sigma_{\mathbf{x}}^{2}$, which joins our analysis in the previous remark.
  \begin{figure}[h]
            \centering
           \includegraphics[height=9cm,width=11.5cm]{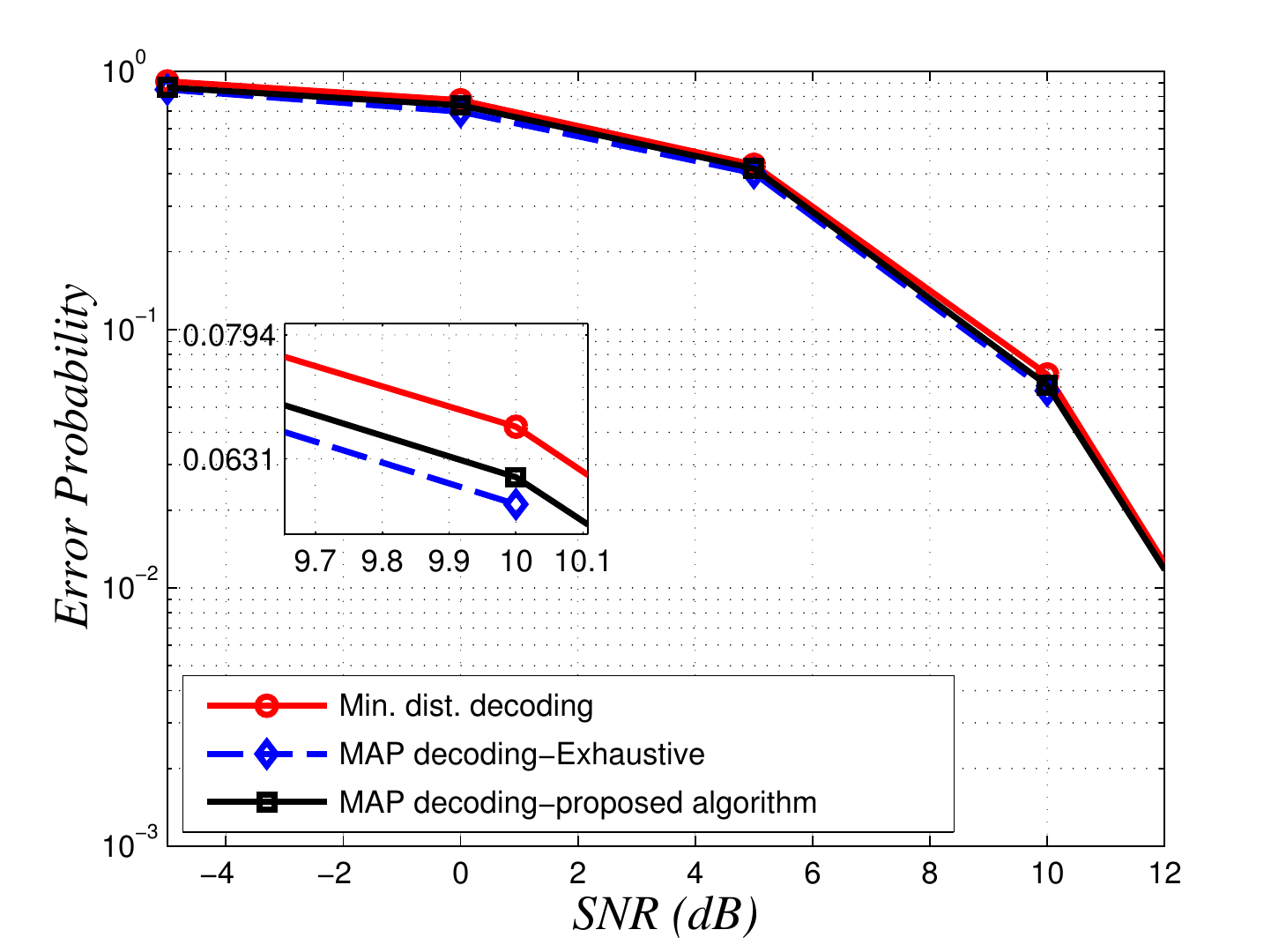}
         \caption{Error performance for $n=2,N=5,P=6.5$.}\label{perf11}
          \end{figure}
          \begin{figure}[h]
              \centering
              \includegraphics[height=9cm,width=11.5cm]{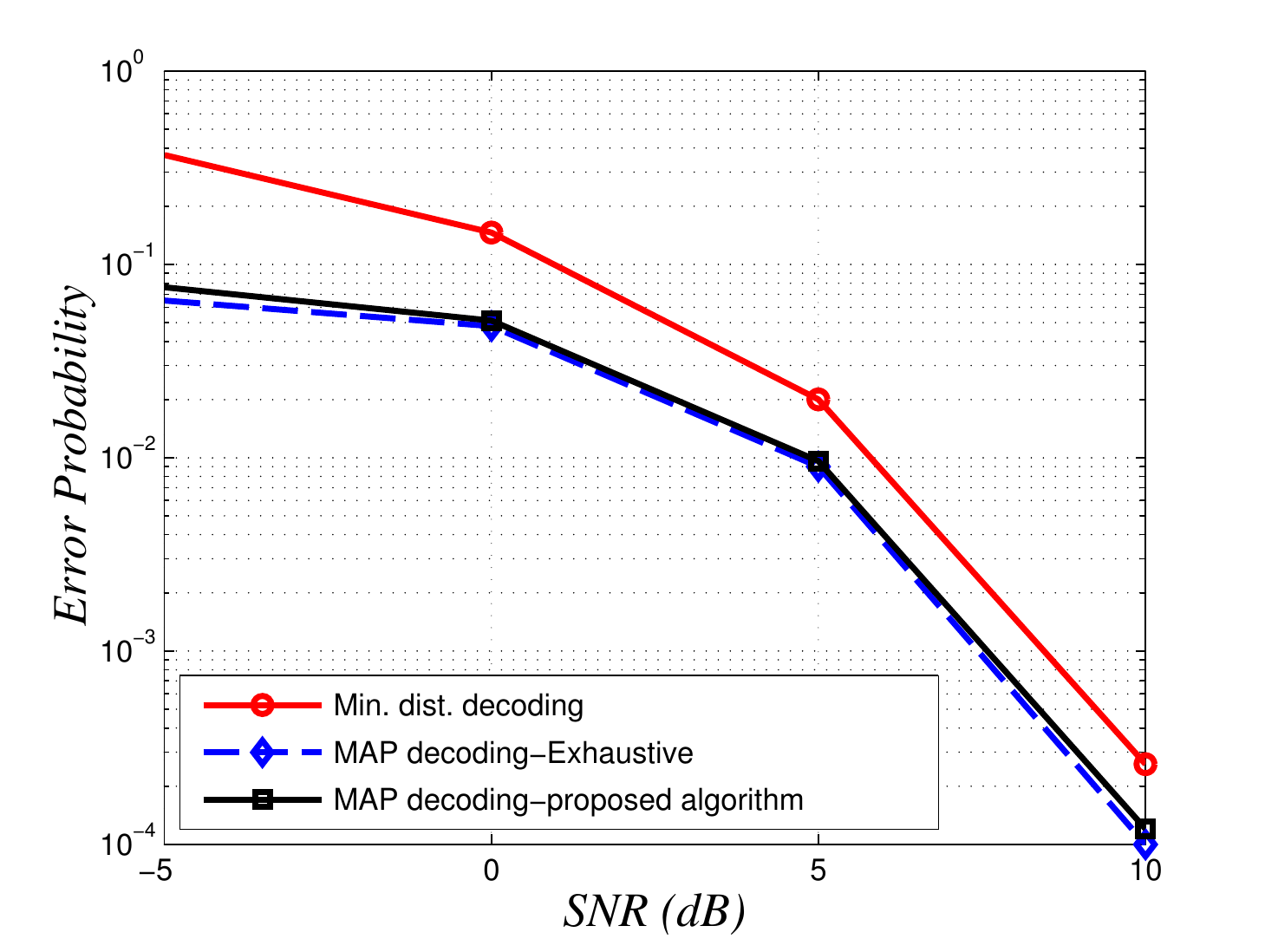}
              \caption{Error performance for $n=4,N=2,P=1$.}\label{perf2map}
              \end{figure}\\
\emph{Example 2: 4-Dimensional lattice $(n=4)$}
In this example we have consider the integer Fine lattice $\Lambda_{\mathrm{F}}$ of a generator matrix the identity $\mathbf{I}_{4}$ together with a cubic shaping region according to $P=1$. The aim of considering this example is to analyze the performance of the MAP decoder when the lattice dimension increases. Simulation results depicted in Fig.\ref{perf2map} show that our proposed MAP algorithm allows to achieve a gain of $1$dB at a codeword error rate of $10^{-3}$ over the minimum distance decoder while keeping a small gap to the exhaustive search. This case shows the merit of applying the MAP decoding in settings where the product $N\sigma_{\mathbf{x}}^{2}$ is small. In addition, we notice that the gap between the MAP decoder and the conventional one is independent of the lattice dimension, it rather increases in settings involving small $N\sigma_{\mathbf{x}}^{2}$.

\section{Conclusion}
This work was dedicated to decoding aspects for the Compute-and-Forward protocol in the basic multiple access real-valued channel. In a first part we studied the fading channel case assuming integer-valued lattices. We analyzed the $n-$dimensional case and proposed a novel near-ML decoder based on diophantine approximation. Numerical results for the $1-$D scenario show the gain of this method over the conventional CF decoder at high SNR range. Moreover, we addressed the Gaussian channel case. After analyzing the MAP decoding rule, we derived a novel decoding metric and developed practical algorithms based on lattice spherical decoding showed to outperform the standard minimum distance decoder. In future works, we aim to investigate the information theoretic performance of our MAP decoder to evaluate the achievable rate in the two-way Gaussian relay channel.


\appendices
\section{Appendix: MMSE-GDFE preprocessing filters}
\label{appendix:mmsegdfe}
We aim to show that the matrices $\mathbf{F}$ and $\mathbf{B}$ in the equivalent MAP decoding metric correspond respectively to the forward and backward filters of the MMSE-GDFE preprocessing in the channel $\mathbf{y}= \lambda_{\mathrm{s}}+\mathbf{z}$ with input $\lambda_{\mathrm{s}}$ such that $\frac{1}{n} \mathbb{E} \left( \parallel \lambda_{\mathrm{s}} \parallel ^{2} \right)= \sigma_{\mathrm{s}}^{2}$. Let $\mathbf{F}_{\mathrm{m}}$ and $\mathbf{B}_{\mathrm{m}}$ be the filters of the MMSE-GDFE preprocessing such that:
 $\mathbf{F}_{\mathrm{m}} \mathbf{y}= \mathbf{F}_{\mathrm{m}} \lambda_{\mathrm{s}} + \mathbf{F}_{\mathrm{m}} \mathbf{z}= \mathbf{B}_{\mathrm{m}} \lambda_{\mathbf{s}} + \left( \mathbf{F}_{\mathrm{m}} - \mathbf{B}_{\mathrm{m}} \right) \lambda_{\mathrm{s}} + \mathbf{F}_{\mathrm{m}} \mathbf{z}$. 
 Let $ \mathbf{w}= \left( \mathbf{F}_{\mathrm{m}} - \mathbf{B}_{\mathrm{m}} \right) \lambda_{\mathrm{s}} + \mathbf{F}_{\mathrm{m}} \mathbf{z}$ be the effective noise. The MMSE-GDFE filters correspond to the minimization of the variance $\varepsilon$ of the effective noise given by:
 \begin{align}
 \varepsilon &= \frac{1}{n} \mathbb{E} \left[ \mathbf{w}^{t}\mathbf{w} \right] = \frac{1}{n} \mathbb{E} \left[ \mathrm{tr} \left( \mathbf{w} \mathbf{w}^{t} \right)\right]= \frac{1}{n} \mathrm{tr} \left(  \mathbb{E} \left[ \left( \mathbf{F}_\mathrm{m}-\mathbf{B}_{\mathrm{m}} \right) \lambda_{\mathrm{s}} \lambda_{\mathrm{s}}^{t} \left( \mathbf{F}_{\mathrm{m}} - \mathbf{B}_{\mathrm{m}} \right)^{t} \right] + \mathbb{E} \left[ \mathbf{F}_{\mathrm{m}} \mathbf{z} \mathbf{z}^{t} \mathbf{F}_{\mathrm{m}}^{t} \right] \right) \notag \\
 &= \frac{1}{n} \mathrm{tr} \left( \left( \mathbf{F}_\mathrm{m}-\mathbf{B}_{\mathrm{m}} \right) \underbrace{\mathbb{E}\left[ \lambda_{\mathrm{s}} \lambda_{\mathrm{s}}^{t} \right]}_{\sigma_{\mathrm{s}}^{2}\mathbf{I}_n} \left( \mathbf{F}_{\mathrm{m}} - \mathbf{B}_{\mathrm{m}} \right)^{t} + \mathbf{F}_{\mathrm{m}} \underbrace{\mathbb{E}\left[\mathbf{z} \mathbf{z}^{t} \right]}_{\sigma^{2}\mathbf{I}_n} \mathbf{F}_{\mathrm{m}}^{t} \right) \notag \\
 &= \frac{\sigma_{\mathrm{s}}^{2}}{n} \mathrm{tr} \left( \mathbf{F}_{\mathrm{m}} \left( \mathbf{I}_n + \beta^{2}\mathbf{I}_n \right)\mathbf{F}_{\mathrm{m}}^{t}\mathbf{F}_{\mathrm{m}}\mathbf{B}_{\mathrm{m}}^{t}-\mathbf{B}_{\mathrm{m}}\mathbf{F}_{\mathrm{m}}^{t}+\mathbf{B}_{\mathrm{m}}\mathbf{B}_{\mathrm{m}}^{t}  \right) \notag
 \end{align} 
Let the matrix $\mathbf{T}$ such that $\mathbf{T}\mathbf{T}^{t}=(1+\beta^{2})\mathbf{I}_n$ and $\mathbf{G}$ such that $\mathbf{G}=\mathbf{F}_{\mathrm{m}}\mathbf{T}$, then $\varepsilon$ is equal to:
\begin{align}
\varepsilon &= \frac{\sigma_{\mathrm{s}}^{2}}{n} \mathrm{tr} \left( \left( \mathbf{G}-\mathbf{B}_{\mathrm{m}}\mathbf{T}^{-t} \right)\left( \mathbf{G}^{t}-\mathbf{T}^{-1}\mathbf{B}_{\mathrm{m}}^{t} \right) + \mathbf{B}_{\mathrm{m}} \left( \mathbf{I}_n - \left( \mathbf{T} \mathbf{T}^{t} \right)^{-1} \right)\mathbf{B}_{\mathrm{m}}^{t} \right) \notag \\
&= \frac{\sigma_{\mathrm{s}}^{2}}{n} \mathrm{tr} \left( \left( \mathbf{G}-\mathbf{B}_{\mathrm{m}}\mathbf{T}^{-t} \right)\left( \mathbf{G}^{t}-\mathbf{T}^{-1}\mathbf{B}_{\mathrm{m}}^{t} \right) +\frac{\beta^{2}}{1+\beta^{2}} \mathbf{B}_{\mathrm{m}}\mathbf{B}_{\mathrm{m}}^{t} \right) \notag
\end{align}
For fixed $\mathbf{B}_{\mathrm{m}}$ we seek first the optimal forward matrix $\mathbf{F}_{\mathrm{m}}$ which minimizes $\varepsilon$. This corresponds to have $\mathbf{G}=\mathbf{B}\mathbf{T}^{-t}$ which results in: $\mathbf{F}_{\mathrm{m}}=\frac{1}{1+\beta^{2}}\mathbf{B}_{\mathrm{m}}$. We get:
\begin{align}
\varepsilon_{\mathrm{min}} &= \frac{\sigma_{\mathrm{s}}^{2}}{n} \frac{\beta^{2}}{1+\beta^{2}}\mathrm{tr} \left( \mathbf{B}_{\mathrm{m}}\mathbf{B}_{\mathrm{m}}^{t} \right) = \frac{\sigma_{\mathrm{s}}^{2}}{n} \frac{\beta^{2}}{1+\beta^{2}}\mathrm{tr} \left( \mathbf{B}_{\mathrm{m}}^{t}\mathbf{B}_{\mathrm{m}} \right)
\end{align}
We have $\mathbf{B}_{\mathrm{m}}^{t}\mathbf{B}_{\mathrm{m}}=\left(1+\beta^{2} \right)\mathbf{I}_n$ wich leads to the minimum variance: $\varepsilon_{\mathrm{min}}= \sigma_{\mathrm{s}}^{2}\beta^{2}$. Now, we will show that $\mathbf{F}=\mathbf{F}_\mathrm{m}$ and $\mathbf{B}=\mathbf{B}_{\mathrm{m}}$. First, $\mathbf{F}$ and $\mathbf{B}$ satisfy same constraints as the MMSE-GDFE filters. The last piece to prove the equivalence then is to prove that $\mathbf{F}$ and $\mathbf{B}$ allow to minimize the variance of the effective noise $\mathbf{w}$. We compute the corresponding variance refered to $\varepsilon_{eq}$:
\begin{align}
 \varepsilon_{eq} &= \frac{\sigma_{\mathrm{s}}^{2}}{n} \mathrm{tr} \left( \left( \mathbf{F} - \mathbf{B} \right)\left( \mathbf{F} - \mathbf{B} \right)^{t} +  \beta^{2} \mathbf{F} \mathbf{F} \right)= \frac{\sigma_{\mathrm{s}}^{2}}{n} \mathrm{tr} \left( \left(1+\beta^{2} \right) \mathbf{F}\mathbf{F}^{t}-\mathbf{F}\mathbf{B}^{t}-\mathbf{B}\mathbf{F}^{t}+\mathbf{B}\mathbf{B}^{t} \right) \notag \\
 &\stackrel{(a)}{=} \frac{\sigma_{\mathrm{s}}^{2}}{n} \left( (1+\beta^{2})\mathrm{tr}\left(\mathbf{F}\mathbf{F}^{t} \right) - \mathrm{tr}\left(\mathbf{F}\mathbf{B}^{t} \right)- \mathrm{tr}\left(\mathbf{B}\mathbf{F}^{t} \right)+ \mathrm{tr}\left( \mathbf{B} \mathbf{B}^{t} \right) \right) \notag \\
  &\stackrel{(b)}{=} \frac{\sigma_{\mathrm{s}}^{2}}{n} \left( (1+\beta^{2})\mathrm{tr}\left(\mathbf{F}^{t} \mathbf{F} \right) - \mathrm{tr}\left(\mathbf{B}^{t} \mathbf{F} \right)- \mathrm{tr}\left(\mathbf{F}^{t} \mathbf{B} \right)+ \mathrm{tr}\left(\mathbf{B}^{t} \mathbf{B}  \right) \right) \notag \\
 &\stackrel{(c)}{=} \frac{\sigma_{\mathrm{s}}^{2}}{n} \left(  (1+\beta^{2})\mathrm{tr}\left(\mathbf{F}^{t} \mathbf{F} \right)-2\underbrace{\mathrm{tr}\left(\mathbf{F}^{t} \mathbf{B} \right)}_{n}+ \underbrace{\mathrm{tr}\left(\mathbf{B}^{t} \mathbf{B}  \right)}_{(1+\beta^{2})n} \right)= \frac{\sigma_{\mathrm{s}}^{2}}{n} \left( (1+\beta^{2})\mathrm{tr}\left(\mathbf{F}^{t} \mathbf{F} \right)+(\beta^{2}-1)n \right) \notag
 \end{align} 
where (a) follows from linearity of trace, (b) follows from commutativity of trace of matrices $(\mathrm{tr}(\mathbf{AB})=\mathrm{tr}(\mathbf{BA}))$, (c) follows using $\mathrm{tr}(\mathbf{A})=\mathrm{tr}(\mathbf{A}^{t})$. Finally, we use the relation $\mathbf{F}^{t}\mathbf{B}=\mathbf{I}_n$ to deduce that $\mathbf{F}^{t}\mathbf{F}=\left(\mathbf{B}^{t}\mathbf{B} \right)^{-1}$ which gives $\mathrm{tr}\left(\mathbf{F}^{t} \mathbf{F} \right)=\frac{n}{1+\beta^{2}}$. We get then: $\varepsilon_{eq} = \sigma_{\mathrm{s}}^{2} \beta^{2}=\varepsilon_{\mathrm{min}}$.

\nocite{*}   
\small{
\bibliographystyle{unsrt}  
\bibliography{bibliography}}

\begin{thebibliography}{10}

\bibitem{Nazer08}
B.~Nazer and M.~Gastpar.
\newblock Compute-and-forward: Harnessing interference with structured codes.
\newblock In {\em Proceedings of ISIT}, pages 772 --776, July 2008.

\bibitem{Silva_isit10}
C.~Feng, D.~Silva, and F.R. Kschischang.
\newblock An algebraic approach to physical-layer network coding.
\newblock In {\em Proceedings of ISIT}, pages 1017 --1021, June 2010.

\bibitem{Asma1}
A.~Mejri and G.~Rekaya.
\newblock Practical physical layer network coding in multi-sources relay
  channels via the compute-and-forward.
\newblock In {\em Proceedings of WCNC}, pages 166--171, April 2013.

\bibitem{Asma2}
A.~Mejri and G.~Rekaya.
\newblock Bidirectional relaying via network coding: Design algorithm and
  performance evaluation.
\newblock In {\em Proceedings of ICT}, pages 1--5, May 2013.

\bibitem{AsmaGhaya}
A.~Mejri, G.~Rekaya, and J.~C Belfiore.
\newblock Lattice decoding for the compute-and-forward protocol.
\newblock In {\em Proceedings of International Conference on Communications and
  Networking}, pages 1--8, March 2012.

\bibitem{Belfiore12}
J-C. Belfiore and C.~Ling.
\newblock The flatness factor in lattice network coding: Design criterion and
  decoding algorithm.
\newblock In {\em International Zurich Seminar on Communications}, 2012.

\bibitem{Viterbo99}
E.~Viterbo and J.~Boutros.
\newblock A universal lattice code decoder for fading channels.
\newblock {\em IEEE Transactions on Information Theory}, 45(5):1639--1642,
  1999.

\bibitem{Cohen93}
H.~Cohen.
\newblock {\em A course in Computational Algebraic Number Theory}.
\newblock Springer-Verlag, New York, USA, 1993.

\bibitem{Lazebnik96}
F.~Lazebnik.
\newblock {\em On Systems of Linear Diophantine Equations}, volume~69.
\newblock Mathematics Magazine, 1996.

\bibitem{Cormen09}
T.H. Cormen, C.~E. Leiserson, R.~L. Rivest, and C.~Stein.
\newblock {\em Introduction to Algorithms}.
\newblock The MIT Press, 2009.

\bibitem{Clarkson97}
I.V.L. Clarkson.
\newblock {\em Approximation of Linear Forms by Lattice Points with
  Applications to Signal Processing}.
\newblock Ph.D Thesis dissertation, 1997.

\bibitem{Cassel57}
J.W.S. Cassel.
\newblock {\em An Introduction to Diophantine Approximation}.
\newblock Cambridge University Press, 1957.

\bibitem{Behnamfar03}
F.~Behnamfar, F.~Alajaji, and T.~Linder.
\newblock Performance analysis of map decoded space-time orthogonal block codes
  for non-uniform sources.
\newblock In {\em Proceedings of the IEEE Information Theory Workshop}, pages
  46--49, 2003.

\bibitem{ErezZamir}
U.~Erez, S.~Litsyn, and R.~Zamir.
\newblock Lattices which are good for (almost) everything.
\newblock In {\em Proceedings of ITW}, pages 271 -- 274, March 2003.

\bibitem{Forney00}
G.D. Forney, M.D. Trott, and S-Y. Chung.
\newblock Sphere-bound-achieving coset codes and multilevel coset codes.
\newblock {\em IEE Trans. on IT}, 46(3):820--850, 2000.

\bibitem{Banaszczyk93}
W.~Banaszczyk.
\newblock New bounds in some transference theorems in the geometry of numbers.
\newblock {\em Math. Ann.}, 296:625--635, 1993.

\bibitem{Micciancio04}
D.~Micciancio and O.~Regev.
\newblock Worst-case to average-case reductions based on gaussian measure.
\newblock In {\em Proceedings of the 45rd annual symposium on foundations of
  computer science}, pages 371--381, Italy, 2004.

\end{thebibliography}

\end{document}